\PassOptionsToPackage{unicode}{hyperref}
\PassOptionsToPackage{hyphens}{url}

\documentclass[
  11pt]{article}
\usepackage{xcolor}
\usepackage{amsmath,amssymb}
\usepackage{iftex}
\ifPDFTeX
  \usepackage[T1]{fontenc}
  \usepackage[utf8]{inputenc}
  \usepackage{textcomp} 
\else 
  \usepackage{unicode-math} 
  \defaultfontfeatures{Scale=MatchLowercase}
  \defaultfontfeatures[\rmfamily]{Ligatures=TeX,Scale=1}
\fi
\usepackage{lmodern}
\ifPDFTeX\else

\fi

\IfFileExists{upquote.sty}{\usepackage{upquote}}{}
\IfFileExists{microtype.sty}{
  \usepackage[]{microtype}
  \UseMicrotypeSet[protrusion]{basicmath} 
}{}
\makeatletter
\@ifundefined{KOMAClassName}{
  \IfFileExists{parskip.sty}{%
    \usepackage{parskip}
  }{
    \setlength{\parindent}{0pt}
    \setlength{\parskip}{6pt plus 2pt minus 1pt}}
}{
  \KOMAoptions{parskip=half}}
\makeatother

\makeatletter
\ifx\paragraph\undefined\else
  \let\oldparagraph\paragraph
  \renewcommand{\paragraph}{
    \@ifstar
      \xxxParagraphStar
      \xxxParagraphNoStar
  }
  \newcommand{\xxxParagraphStar}[1]{\oldparagraph*{#1}\mbox{}}
  \newcommand{\xxxParagraphNoStar}[1]{\oldparagraph{#1}\mbox{}}
\fi
\ifx\subparagraph\undefined\else
  \let\oldsubparagraph\subparagraph
  \renewcommand{\subparagraph}{
    \@ifstar
      \xxxSubParagraphStar
      \xxxSubParagraphNoStar
  }
  \newcommand{\xxxSubParagraphStar}[1]{\oldsubparagraph*{#1}\mbox{}}
  \newcommand{\xxxSubParagraphNoStar}[1]{\oldsubparagraph{#1}\mbox{}}
\fi
\makeatother

\usepackage{longtable,booktabs,array}
\usepackage{calc} 

\usepackage{etoolbox}
\makeatletter
\patchcmd\longtable{\par}{\if@noskipsec\mbox{}\fi\par}{}{}
\makeatother

\IfFileExists{footnotehyper.sty}{\usepackage{footnotehyper}}{\usepackage{footnote}}
\makesavenoteenv{longtable}
\usepackage{graphicx}
\makeatletter
\newsavebox\pandoc@box
\newcommand*\pandocbounded[1]{
  \sbox\pandoc@box{#1}%
  \Gscale@div\@tempa{\textheight}{\dimexpr\ht\pandoc@box+\dp\pandoc@box\relax}%
  \Gscale@div\@tempb{\linewidth}{\wd\pandoc@box}%
  \ifdim\@tempb\p@<\@tempa\p@\let\@tempa\@tempb\fi
  \ifdim\@tempa\p@<\p@\scalebox{\@tempa}{\usebox\pandoc@box}%
  \else\usebox{\pandoc@box}%
  \fi%
}

\def\fps@figure{htbp}
\makeatother

\providecommand{\tightlist}{%
  \setlength{\itemsep}{0pt}\setlength{\parskip}{0pt}}

\usepackage[]{natbib}
\usepackage{bbm}
\usepackage[nomarginpar]{geometry}
\usepackage{setspace}
\usepackage{lscape}
\usepackage{rotating}
\usepackage{tabularx, calc}
\usepackage{threeparttable}
\usepackage{adjustbox}
\usepackage{ragged2e}
\usepackage{changepage}
\usepackage{authblk}
\usepackage{caption}
\makeatletter
\@ifpackageloaded{caption}{}{\usepackage{caption}}
\AtBeginDocument{%
\ifdefined\contentsname
  \renewcommand*\contentsname{Table of contents}
\else
  \newcommand\contentsname{Table of contents}
\fi
\ifdefined\listfigurename
  \renewcommand*\listfigurename{List of Figures}
\else
  \newcommand\listfigurename{List of Figures}
\fi
\ifdefined\listtablename
  \renewcommand*\listtablename{List of Tables}
\else
  \newcommand\listtablename{List of Tables}
\fi
\ifdefined\figurename
  \renewcommand*\figurename{Figure}
\else
  \newcommand\figurename{Figure}
\fi
\ifdefined\tablename
  \renewcommand*\tablename{Table}
\else
  \newcommand\tablename{Table}
\fi
}
\@ifpackageloaded{float}{}{\usepackage{float}}
\floatstyle{ruled}
\@ifundefined{c@chapter}{\newfloat{codelisting}{h}{lop}}{\newfloat{codelisting}{h}{lop}[chapter]}
\floatname{codelisting}{Listing}

\usepackage{amsthm}
\theoremstyle{plain}
\newtheorem{theorem}{Theorem}[section]
\theoremstyle{plain}
\newtheorem{corollary}{Corollary}[section]
\theoremstyle{plain}
\newtheorem{lemma}{Lemma}[section]
\theoremstyle{remark}
\AtBeginDocument{}

\newtheorem{refremark}{Remark}[section]

\makeatother
\makeatletter
\@ifpackageloaded{caption}{}{\usepackage{caption}}
\@ifpackageloaded{subcaption}{}{\usepackage{subcaption}}
\makeatother
\makeatletter
\@ifpackageloaded{amsthm}{}{\usepackage{amsthm}}
\makeatother
\makeatletter
\@ifundefined{c@lemma}{}{\let\c@lemma\c@theorem}
\makeatother
\makeatletter
\@ifundefined{c@corollary}{}{\let\c@corollary\c@theorem}
\makeatother
\makeatletter
\@ifundefined{c@proposition}{}{\let\c@proposition\c@theorem}
\makeatother
\makeatletter
\@ifundefined{c@conjecture}{}{\let\c@conjecture\c@theorem}
\makeatother
\makeatletter
\@ifundefined{c@definition}{}{\let\c@definition\c@theorem}
\makeatother
\makeatletter
\@ifundefined{c@example}{}{\let\c@example\c@theorem}
\makeatother
\makeatletter
\@ifundefined{c@exercise}{}{\let\c@exercise\c@theorem}
\makeatother
\theoremstyle{remark}
\makeatletter
\@ifundefined{c@theorem}{
  \newtheorem{myexm}{Example}[section]
  \let\c@theorem\c@myexm
  \let\thetheorem\themyexm
  \let\p@theorem\p@myexm
}{
  \newtheorem{myexm}[theorem]{Example}
}
\makeatother
\theoremstyle{plain}
\newtheorem{mylem}[theorem]{Lemma}
\theoremstyle{plain}
\newtheorem{myprp}[theorem]{Proposition}
\theoremstyle{plain}

\theoremstyle{definition}
\newtheorem{mydef}[theorem]{Definition}
\theoremstyle{definition}
\newtheorem{asm}[theorem]{Assumption}
\theoremstyle{definition}
\newtheorem{axm}[theorem]{Axiom}
\theoremstyle{plain}
\newtheorem{mycor}[theorem]{Corollary}
\usepackage{bookmark}
\IfFileExists{xurl.sty}{\usepackage{xurl}}{} 
\hypersetup{
  pdftitle={Heterogeneous Elasticities, Aggregation, and Retransformation Bias},
  hidelinks,
  pdfcreator={LaTeX via pandoc}}

\begin{document}
\date{}
\title{Heterogeneous Elasticities, Aggregation, and Retransformation
Bias}
\renewcommand\Authfont{\large}
\renewcommand\Affilfont{\large}
\setlength{\affilsep}{1em}
\author[1]{Ellen Munroe}
\author[1]{Alexander Newton}
\author[1]{Meet Shah}
\affil[1]{Economics Department, London School of Economics and Political
Science}
\maketitle

\begin{abstract}
Economists often interpret estimates from linear regressions with log
dependent variables as elasticities. However, the coefficients from
log-log regressions estimate the elasticity of the geometric mean of
\(y_i|x_i\) , not the arithmetic mean. The unbounded difference between
the two is known as retransformation bias and can take either sign. We
develop a specification-robust debiased estimator of the average
arithmetic elasticity and re-estimate 50 results from top 5 papers
published in 2020. We find that 19 are significantly different, with the
median absolute difference being 65\% of the OLS elasticity estimate.
Furthermore, we show standard instrumental variables assumptions with
log dependent variables do not identify the elasticity. We specify a
control function approach and re-estimate papers that use 2SLS with log
dependent variables. We find that 13 of 19 results from top 5 papers are
significantly different between the two approaches. Retransformation
bias arises as a result of heterogeneous responses. The geometric mean
elasticity corresponds to the average response. Arithmetic and geometric
means are elements of the power mean family. We show power mean
elasticities are sufficient statistics for a common class of decision
problems.
\end{abstract}

\begingroup
\renewcommand\thefootnote{}\footnotetext{We thank Kimia Zargarzadeh for
excellent research assistance. We thank STICERD for financial support.}
\endgroup
\setcounter{footnote}{0}

\newpage

\section{Introduction}\label{sec-intro}

Economists are often interested in the arithmetic mean elasticity
\(\frac{\partial\log\mathbb{E}[y|x]}{\partial \log x}\).\footnote{Almost
  all discussion in this paper applies to semi-elasticities as it does
  to elasticities - save for the decision problem axiomatisation in
  Section~\ref{sec-discussion}. Our estimators work for
  semi-elasticities as they do for elasticities.} This quantity is used
to understand how, for example, quantity sold changes with price,
government revenue changes with tax rates, carbon emissions change with
carbon taxes, or aggregate welfare changes with policy.

To estimate this object, economists use a log-log regression. However,
as \citet{goldbergerInterpretationEstimationCobbDouglas1968} notes, the
estimated parameters in these regressions converge to the conditional
\textbf{geometric} mean elasticity
\(\frac{\partial\mathbb{E}[\log y|x]}{\partial \log x}\)- not the
arithmetic mean elasticity.

To demonstrate, suppose we estimate the log-log model \[
\log{y_{i}} = \beta \log x_{i} + v_{i}, \quad \mathbb{E}[v_{i}|x_{i}] = 0 \tag{1}
.\] The geometric mean elasticity will be \(\beta\), whereas the
arithmetic mean elasticity will be
\(\beta + \frac{\partial \log E[\exp(v_{i})|x_{i}]}{\partial \log x_{i}}\).
The second term will only be zero if the error, \(v_{i}\), and the
regressor, \(x_{i}\), are fully independent - even heteroskedasticity
will introduce a wedge between the two. This may be of any sign and
size. Thus, any consistent estimator \(\hat{\beta}\) for \(\beta\) will
in general be biased for the arithmetic mean elasticity. Following
\citet{duanSmearingEstimateNonparametric1983}, we call this
retransformation bias.

We find the arithmetic mean elasticity and the geometric mean elasticity
are frequently different in practice. Using our specification-robust
debiased machine learning estimator
(Section~\ref{sec-debiased-estimators}), we re-estimate 50 results from
top 5 papers published in 2020. We find that 19 have estimated average
arithmetic mean elasticities that are statistically significantly
different from the geometric mean elasticities estimated by OLS. Amongst
these results, we find the median absolute bias to be 65\% of the OLS
elasticity estimate. Four had opposite signs.

In previous work, \citet{santossilvaLogGravity2006} propose using PPML
to estimate the arithmetic mean elasticity for trade models. This is
consistent under a different model:

\[
y_{i} = x_{i}^{\gamma} \exp(u_{i}), \quad \mathbb{E}[\exp{ u_{i} }|x_{i}] = 1 \tag{2}.
\]

Here, \(\gamma\) is the arithmetic mean elasticity and PPML estimates it
consistently. However, direct comparisons between log-log OLS and PPML
estimates are flawed as the estimands are simply different. Furthermore,
under heterogeneous responses PPML does not estimate the arithmetic mean
elasticity. Indeed, with heterogeneous elasticities only the geometric
mean elasticity is constant. Under heterogeneity model (1) - the log-log
model - is the most natural statistical model.

If the log-log model, with constant geometric mean elasticity, is the
most natural why might economists want to estimate the
\textbf{arithmetic} mean elasticity? Using the potential outcomes
framework we provide monopolist and social planner examples which give
intuition for when heterogeneous responses require the use of our
estimator to estimate the arithmetic mean elasticity. We extend our
discussion to power mean elasticities in general and motivate their
estimation as the sufficient statistics for a general class of decision
problem.

When economists have endogeneity concerns they will often use
instrumental variables to estimate log-log structural equations. We show
an impossibility result for identification of the arithmetic mean
elasticity under standard IV assumptions. Identifying this object
requires the additional assumption of triangularity to enable a control
function approach. We develop a separate debiased machine learning
estimator with instrumental variables for the average arithmetic mean
elasticity and re-estimate results from the literature.

The paper proceeds as follows. Section~\ref{sec-intro} concludes with a
motivating example from public finance and a literature review.
Section~\ref{sec-model} presents a heterogeneous elasticities model,
derives a wedge between geometric and arithmetic means represented as
the retransformation bias term, and generalizes to wedges between other
power means. Section~\ref{sec-discussion} discusses aggregation in
economic models, provides axioms under which decision-makers would
estimate power mean elasticities, and shows how aggregation choice maps
to welfare weights in a social planner's problem.
Section~\ref{sec-identification} discusses how to obtain arithmetic mean
elasticities from RCTs and instrumental variables (IV) analysis, showing
that recovering the arithmetic mean elasticity from log-log IV
regressions is impossible under standard assumptions.
Section~\ref{sec-debiased-estimators} presents our estimators for
average arithmetic mean elasticities under regular identification
assumptions and for IV analysis. Section~\ref{sec-empirical}
re-estimates 50 papers from Top 5 economics journals, showing that the
retransformation bias term is sizeable in practice.
Section~\ref{sec-conclusion} concludes.

\subsection{Motivating Example: Marginal Value of Public
Funds}\label{motivating-example-marginal-value-of-public-funds}

Public economics often uses the marginal value of public funds (MVPF) to
evaluate policy impacts. The MVPF is a ratio of the beneficiaries'
willingness to pay and the net government cost. For evaluations of top
marginal tax rate changes, \citet{hendrenUnifiedWelfareAnalysis2020}
show that MVPF can be given by \[
MVPF = \frac{1}{1 - \frac{\tau}{1-\tau} z \epsilon_{eti}}
\] where, \(\tau\) is the tax rate, \(z\) is the pareto parameter for
the distribution of incomes Y and \[
\epsilon_{eti} = \frac{d\mathbb{E}[{Y}]}{d(1-\tau)} \cdot \frac{1-\tau}{\mathbb{E}[Y]}
\] is the elasticity of the arithmetic mean of top incomes with respect
to the after-tax keep rate.

To reach an estimate of the MVPF, their paper relies on elasticity
estimates obtained from log-log regressions, for example from
\citet{saezEffectMarginalTax2003}. As we show, these will estimate the
geometric mean elasticity of top incomes with respect to the after-tax
keep rate, which is also equal to the average elasticity under
heterogeneity. The average elasticity can be arbitrarily different from
the elasticity of the arithmetic mean.

To provide intuition, consider an economy with two individuals at the
top tax rate, individual 1 has an income of 1000\$ and an elasticity of
0.3, individual 2 has an income of 10000 dollars and an elasticity of
0.1. The average elasticity is 0.2, but the elasticity of the arithmetic
mean is (300+1000)/11000 = 0.11. In particular, to obtain the required
elasticity, one needs to weigh the elasticities of those with higher
incomes more than those with lower incomes, even within top income
brackets.

\subsection{Related literature}\label{related-literature}

We see our paper as uniting the statistical property of retransformation
bias with economic models of heterogenous elasticities and aggregation.

\citet{lewbelAggregationLogLinearModels1992} studies when individual
heterogenous elasticity agents can be approximated using a
representative agent. We, conversely, focus on situations where they
cannot be modelled as such, and focus on the properties of different
estimators under heterogenous elasticities.
\citet{breinlichTradeGravityAggregation2024} are closely related. They
study the performance of PPML in trade models when elasticities vary at
the sector level, finding that it approximates a weighted average of the
sectoral elasticities when trade costs do not vary at sector level, but
is biased otherwise. We generalize their results to arbitrary elasticity
distributions, consider the performance of OLS and provide a new
estimator for the average arithmetic mean elasticity.

Heterogenous responses are often modelled using random coefficients
models \citep{wooldridgeFixedEffectsRelatedEstimators2005}. We study
bias arising from non-linearity in the heterogenous coefficients,
specifically in constant individual elasticity models. To our knowledge,
we are the first to use heterogenous coefficient models to provide
economic meaning to the statistical problem of retransformation bias.
\citet{qianTestingOmittedHeterogeneity2026} provides a test for omitted
heterogeneity, and provides an estimator for the average coefficient for
the case when heterogeneity is asymptotically disappearing. We focus on
cases where the heterogeneity distribution is stable, and give examples
for when the relevant estimand is not the average coefficient.

Transformation bias has long been recognized in economics and statistics
literature. \citet{goldbergerInterpretationEstimationCobbDouglas1968}
notes that while transformation bias affects conditional expectations,
it does not affect conditional medians---though the latter requires a
different identifying assumption (zero conditional median of errors
rather than zero conditional mean). We focus on retransformation bias
from the model specification \(\mathbb{E}[\log y|x]= x \beta\).

\citet{duanSmearingEstimateNonparametric1983} provides a non-parametric
bias-corrected estimator for \(\mathbb{E}[y|x]\) under full independence
between errors and regressors.
\citet{manningLoggedDependentVariable1998} extends this to binary
treatments. Our estimator nests Manning's approach while considering a
more comprehensive distributional framework.
\citet{aiSemiparametricDerivativeEstimator2008} develop related
semi-parametric estimators for the conditional expectation and its
derivatives, but our paper builds on their estimator. Our work relies on
debiased machine learning \citep{chernozhukovDoubleDebiasedMachine2018},
which allows us to bypass the curse of dimensionality that their
proposed spline approach faces. The estimator we propose approximates an
unknown, potentially high-dimensional function of the regressors using
flexible machine-learning methods. Such methods typically require
regularization to control overfitting, which in turn slows convergence
and introduces bias. To address this issue, we implement the locally
robust semiparametrics of
\citet{chernozhukovLocallyRobustSemiparametric2022}~in our elasticity
setting. This allows us to remove the leading regularization bias and to
obtain an estimator that is both consistent and asymptotically normal
under their high-level conditions. Our IV estimator uses the automatic
debiasing procedure from
\citet{chernozhukovAutomaticDebiasedMachine2024} to avoid estimation of
complex Riesz representers. For both estimators we use score matching
\citep{hyvarinenEstimationNonnormalizedStatistical2005, hanNeuralNetworkbasedScore2024}.

When \(y\) takes zero values, researchers often use \(\log(1+y)\) or
\(\operatorname{arcsinh}(y)\) transformations.
\citet{chenLogsZerosProblems2024} show that coefficients from these
transformations are sensitive to scaling of \(y\), arising from
difficulties distinguishing extensive and intensive margin effects.
However, they do not consider transformation bias from non-linear
transformations when making inferences about \(\mathbb{E}[y_i|x_i]\). We
only focus on cases where the outcome variable is strictly positive,
relative to their focus on scaling when it can be zero.
\citet{mullahyWhyTransformPitfalls2024} discusses transformation bias in
the \(\log(1+y)\) context as another reason to avoid variable
transformations.

\citet{bellemareElasticitiesInverseHyperbolic2020} discuss when
coefficients from \(\operatorname{arcsinh}\) transformations can be
interpreted as percent changes, but do not address transformation bias
from conditional expectations of transformed errors.
\citet{nortonInverseHyperbolicSine2022} applies
\citet{duanSmearingEstimateNonparametric1983}'s smearing estimator to
arcsinh regressions, but doesn't tackle the case where there is
dependence between regressors and the error term.
\citet{thakralWhenAreEstimates2025} suggest power transformations over
\(\log(1+y)\) or \(\operatorname{arcsinh}\), citing scaling sensitivity.
Our paper demonstrates that transformation bias occurs with any
non-linear transformation.

\citet{santossilvaLogGravity2006} encourage researchers to use Poisson
pseudo-maximum likelihood (PPML) estimation to estimate the parameters
of the exponential model. However, the exponential is a different model
to the log-linear model, and in general they cannot both hold at the
same time. We show that the exponential model is incompatible with each
individual having a heterogeneous, constant elasticity. The log-linear
model allows for non-constant elasticities of \(\mathbb{E}[y_i|x_i]\)
over \(x\). In this paper we show conditions for which both models can
coexist, and show that in most empirical models, one model being
correctly specified implies the other cannot be.

\section{Model}\label{sec-model}

\subsection{Primitives}\label{primitives}

Let \(x \in \mathbb{R}_+\) denote a scalar variable of interest (e.g.~a
price, tax, policy parameter, or cost shifter). For each unit
\(\omega \in \Omega\), let \(Y(x;\omega) > 0\) denote an outcome of
interest, such as demand, output, emissions, or income. Units are
distributed according to a probability measure \(F\). Expectations
\(\mathbb{E}[\cdot]\) are taken with respect to \(F\).

The index \(\omega\) summarizes unit-specific characteristics that are
fixed with respect to \(x\) and unobserved by the econometrician. These
characteristics may reflect preferences, technologies, constraints, or
other sources of persistent heterogeneity.

We allow outcomes to respond proportionally to changes in \(x\), while
permitting heterogeneity in both scale and responsiveness across units.

\begin{asm}[Heterogenous Elasticity Potential Outcomes]\label{asm-randomcoef}

For each \(\omega\), \[
\log Y(x;\omega) = a(\omega) + \varepsilon(\omega)\,\log x.
\] The coefficients \(a(\omega)\) and \(\varepsilon(\omega)\) have
finite moments. This is equivalent to a random coefficients model.

\end{asm}

The term \(a(\omega)\) captures heterogeneity in scale, while
\(\varepsilon(\omega)\) captures heterogeneity in responsiveness to
\(x\). No restrictions are imposed on the joint distribution of
\((a(\omega), \varepsilon(\omega))\). Define the individual elasticity
\[
\varepsilon(\omega) \equiv \frac{d \log Y(x;\omega)}{d \log x}.
\] Under \hyperref[asm-randomcoef]{Assumption~\ref*{asm-randomcoef}},
individual elasticities are constant in \(x\) but heterogeneous across
units. That is, every unit behaves as if it has a constant elasticity,
but that constant elasticity is different between units.

Economic objectives typically depend on aggregate outcomes. We therefore
consider a family of aggregation rules that nests common cases.

\begin{mydef}[Power-mean aggregators]\label{mydef-powermean}

For \(\phi \in \mathbb{R}\), define \[
M_\phi(x) \equiv
\begin{cases}
\left( \mathbb{E}[Y(x)^\phi] \right)^{1/\phi}, & \phi \neq 0, \\
\exp\left( \mathbb{E}[\log Y(x)] \right), & \phi = 0.
\end{cases}
\]

\end{mydef}

Special cases include:

\begin{itemize}
\tightlist
\item
  \(\phi = 1\), the arithmetic mean, \(M_1(x) = \mathbb{E}[Y(x)]\),
\item
  \(\phi \rightarrow 0\), the geometric mean,
  \(M_0(x) = \exp(\mathbb{E}[\log Y(x)])\),
\item
  \(\phi \rightarrow -\infty\), the minimum function,
  \(M_{-\infty}=\min{Y(x)}\).
\end{itemize}

This family provides a continuous mapping between aggregation rules that
emphasize total outcomes and those that emphasize typical outcomes.

For each aggregator \(M_\phi(x)\), define the associated elasticity.

\begin{mydef}[Power-mean elasticity]\label{mydef-powermeanelasticity}

The elasticity of \(M_\phi(x)\) with respect to \(x\) is \[
\varepsilon_\phi(x) \equiv \frac{d \log M_\phi(x)}{d \log x}.
\]

\end{mydef}

Under \hyperref[asm-randomcoef]{Assumption~\ref*{asm-randomcoef}},
power-mean elasticities admit a simple representation.

\begin{mylem}[Power-mean elasticity representer]\label{mylem-represent}

For any \(\phi \in \mathbb{R}\), \[
\varepsilon_\phi(x)
=
\frac{\mathbb{E}_{\omega}\!\left[ Y(x; \omega)^\phi \, \varepsilon(\omega) \right]}
     {\mathbb{E}\!\left[ Y(x; \omega)^\phi \right]}.
\]

\end{mylem}

Thus, \(\varepsilon_\phi(x)\) is a reweighted average of individual
elasticities, where the weights depend on both the aggregation parameter
\(\phi\) and the level of \(x\).

The two common cases we have discussed are immediate,

\begin{itemize}
\item
  \textbf{Geometric mean elasticity (\(\phi \rightarrow 0\))}: \[
  \varepsilon_0(x) = \varepsilon_0 = \mathbb{E}[\varepsilon(\omega)],
  \] which is constant in \(x\).
\item
  \textbf{Arithmetic mean elasticity (\(\phi = 1\))}: \[
  \varepsilon_1(x)
  =
  \frac{\mathbb{E}[Y(x)\,\varepsilon(\omega)]}
     {\mathbb{E}[Y(x)]},
  \] which generally varies with \(x\) as the composition of aggregate
  outcomes changes.
\end{itemize}

The elasticity \(\varepsilon_\phi(x)\) measures the percentage response
of an aggregate outcome defined by the aggregation rule \(\phi\).
Different values of \(\phi\) emphasize different parts of the
cross-sectional distribution. \(\phi \rightarrow 0\) weights units
symmetrically in log space and reflects the response of a typical unit;
\(\phi = 1\) weights units in proportion to their contribution to total
outcomes. Holding elasticities fixed, higher values of \(\phi\) put more
weight on high outcomes, and lower values of \(\phi\) put more weight on
low outcomes.

As shown below, different economic objectives correspond to different
values of \(\phi\), and only a subset of these elasticities are
identified by standard empirical approaches.

\subsection{Results}\label{results}

We begin by formalizing the distinction between elasticities associated
with different aggregation rules.

\begin{myprp}[Non-equivalence under heterogeneity]\label{myprp-noneqhet}

Suppose \hyperref[asm-randomcoef]{Assumption~\ref*{asm-randomcoef}}
holds. If \(\varepsilon(\omega)\) is non-degenerate, then for any
\(\phi \neq 0\), \[
\varepsilon_\phi(x) \neq \varepsilon_0
\] for a set of \(x\) with positive measure. Moreover,
\(\varepsilon_\phi(x)\) is generically a non-constant function of \(x\).

\end{myprp}

\hyperref[myprp-noneqhet]{Proposition~\ref*{myprp-noneqhet}} can also be
stated as, even if every individual is isoelastic, the elasticities of
every mean except the geometric mean are not isoelastic.
\hyperref[myprp-noneqhet]{Proposition~\ref*{myprp-noneqhet}} highlights
a fundamental distinction between elasticities defined by different
aggregation rules. The geometric elasticity \(\varepsilon_0\) averages
individual elasticities symmetrically in log space and is invariant to
changes in \(x\). In contrast, elasticities associated with other
aggregation rules reweight individual responses in a manner that depends
on the level of \(x\).

This reweighting reflects how, as \(x\) changes, units with different
responsiveness contribute differentially to aggregate outcomes. As a
result, aggregate elasticities relevant for totals, profits, or other
non-log objectives generally vary with \(x\) even when individual
elasticities are constant.

This also implies that parameters obtained from log-log OLS cannot
correspond to any elasticity other than that of the geometric mean, as
all the other ones are non-constant in x. Moreover, any estimator which
targets a constant power mean elasticity except geometric cannot be
consistent. For example, PPML targets a constant arithmetic mean
elasticity, therefore under heterogenous elasticities, the PPML
estimator is inconsistent.

One illustrative example is demand, let \(D(p; \omega)\) reflect the
quantity demanded at price \(p\). Then the elasticity of the arithmetic
mean of demand to price is \[
\varepsilon_{1}(p) = \dfrac{\mathbb{E}[D(p; \omega) \varepsilon(\omega)]}{\mathbb{E}[D(p; \omega)]} .
\]

We necessarily weigh those with higher demand more. A high elasticity
individual with high demand affects total demand more than a high
elasticity individual with low demand! Alternatively - for the same
demand, one should weigh a high elasticity individual more than a low
elasticity individual, as they will affect total demand more.

\begin{myprp}[Wedge decomposition of power-mean elasticities]\label{myprp-wedgedecomp}

Suppose \hyperref[asm-randomcoef]{Assumption~\ref*{asm-randomcoef}}
holds and \(\mathbb{E}[Y(x)^\phi] < \infty\) for the \(\phi\) of
interest. Define \[
G_\phi(x) \equiv \mathbb{E}[Y(x)^\phi].
\] Then for any \(\phi \neq 0\), \[
\varepsilon_\phi(x)
=
\frac{1}{\phi}\frac{d \log G_\phi(x)}{d \log x}.
\] Moreover, writing
\(\varepsilon_0 \equiv \mathbb{E}[\varepsilon(\omega)]\) and defining
the centered slope
\(\tilde\varepsilon(\omega)\equiv \varepsilon(\omega)-\varepsilon_0\),
we have the identity \[
\varepsilon_\phi(x)
=
\varepsilon_0
+
\frac{1}{\phi}\frac{d}{d\log x}
\log \mathbb{E}\!\left[\exp\!\left(\phi a(\omega)\right)\,x^{\phi \tilde\varepsilon(\omega)}\right].
\] In particular, for \(\phi=1\) (the arithmetic mean), \[
\varepsilon_1(x)
=
\varepsilon_0
+
\frac{d}{d\log x}
\log \mathbb{E}\!\left[\exp(a(\omega))\,x^{\tilde\varepsilon(\omega)}\right].
\]

\end{myprp}

Proof in Section~\ref{sec-wedgedecompproof}

Since the distribution of \(\omega\) is fixed, we may differentiate
inside the expectation to obtain \[
\varepsilon_\phi(x)-\varepsilon_0
=
\frac{\mathbb{E}\!\left[\exp(\phi a(\omega))\,x^{\phi\tilde\varepsilon(\omega)}\,\tilde\varepsilon(\omega)\right]}
{\mathbb{E}\!\left[\exp(\phi a(\omega))\,x^{\phi\tilde\varepsilon(\omega)}\right]}
\] which is a tilted mean of the difference from average.

\hyperref[myprp-wedgedecomp]{Proposition~\ref*{myprp-wedgedecomp}}
demonstrates that the difference between \(\varepsilon_0\) and
\(\varepsilon_\phi(x)\) requires understanding how the tilted moment
\(\mathbb{E}[Y(x)^\phi]\) varies with \(x\). Log-linear methods identify
\(\varepsilon_0\) because it depends only on the first moment of
\(\log Y(x;\omega)\). By contrast, \(\varepsilon_\phi(x)\) depends on
how heterogeneity is reweighted in levels as \(x\) changes, which is
summarized by the second term in the decomposition.

This is a `tilted' average of the residuals, where those individuals
with higher scales (\(a\)), and larger differences from the average
elasticity, are given more weight.

Focusing on the arithmetic geometric wedge, \[
\varepsilon_1(x)-\varepsilon_0
=
\frac{\mathbb{E}\!\left[\exp(a(\omega))\,x^{\tilde\varepsilon(\omega)}\,\tilde\varepsilon(\omega)\right]}
{\mathbb{E}\!\left[\exp(a(\omega))\,x^{\tilde\varepsilon(\omega)}\right]}.
\]

A first-order expansion around \(x=1\) yields \[
\varepsilon_1(x)-\varepsilon_0
=
\frac{\mathbb{E}\!\left[\exp(a(\omega))\,\tilde\varepsilon(\omega)\right]}{\mathbb{E}\!\left[\exp(a(\omega))\right]}
\;+\;
\log x\left(
\frac{\mathbb{E}\!\left[\exp(a(\omega))\,\tilde\varepsilon(\omega)^2\right]}{\mathbb{E}\!\left[\exp(a(\omega))\right]}
-
\left(\frac{\mathbb{E}\!\left[\exp(a(\omega))\,\tilde\varepsilon(\omega)\right]}{\mathbb{E}\!\left[\exp(a(\omega))\right]}\right)^2
\right)
\;+\; o(\log x).
\]

The first term captures sorting between scale \(\exp(a(\omega))\) and
responsiveness \(\tilde\varepsilon(\omega)\) at \(x=1\). One can think
of this, in a demand scenario, as describing the covariance between the
demand shock and elasticity. If demand shocks are larger for those with
higher elasticities, then, this term will be positive, else it will be
negative. The coefficient on \(\log x\) is the (scale-weighted) variance
of \(\tilde\varepsilon(\omega)\) and is therefore weakly positive, with
strict positivity whenever \(\tilde\varepsilon(\omega)\) is
non-degenerate.

The tilted-mean representation immediately delivers a sharp condition
under which the arithmetic and geometric elasticities coincide.

\begin{mycor}[No Wedge]\label{mycor-nowedge}

Under Assumption 1, \[
\varepsilon_1(x)=\varepsilon_0 \ \text{for all } x
\quad\Longleftrightarrow\quad
\mathbb{E}\!\left[\exp(a(\omega))\,x^{\tilde\varepsilon(\omega)}\,\tilde\varepsilon(\omega)\right]=0 \ \text{for all } x,
\] where
\(\tilde\varepsilon(\omega)\equiv \varepsilon(\omega)-\varepsilon_0\).
In particular, a sufficient condition is
\(\tilde\varepsilon(\omega)\equiv 0\) almost surely (no heterogeneity in
elasticities), in which case \(\varepsilon_\phi(x)=\varepsilon_0\) for
all \(\phi\) and all \(x\).

\end{mycor}

\begin{myexm}[Closed-form power-mean elasticities under Gaussianity]\label{myexm-gaussian}

This section provides a tractable special case that illustrates how the
entire family of power-mean elasticities varies with \(x\).

Assume \(a(\omega)\equiv a\) is constant and
\(\varepsilon(\omega)\sim \mathcal{N}(\bar\varepsilon,\sigma^2)\). Then,
for any \(\phi\in\mathbb{R}\), \[
M_\phi(x)
=
\exp\!\left(a+\bar\varepsilon\log x+\tfrac{1}{2}\phi\sigma^2(\log x)^2\right),
\] and the associated power-mean elasticity is \[
\varepsilon_\phi(x)
\equiv
\frac{d\log M_\phi(x)}{d\log x}
=
\bar\varepsilon+\phi\sigma^2\log x.
\] Two benchmarks follow immediately: \[
\varepsilon_0=\bar\varepsilon,
\qquad
\varepsilon_1(x)=\bar\varepsilon+\sigma^2\log x.
\] Moreover, the full family admits the affine representation \[
\varepsilon_\phi(x)=(1-\phi)\varepsilon_0+\phi\,\varepsilon_1(x),
\] so that \(\phi\) linearly indexes a continuum between the constant
geometric elasticity and the \(x\)-varying arithmetic elasticity.

\end{myexm}

A natural estimator in this model under Gaussianity is to use PPML to
estimate a regression with \(\log{x}\) and \((\log{x})^2\), and combine
the obtained parameters to get the arithmetic mean elasticity.

The results in Section~\ref{sec-model} have shown that, under
heterogeneous responsiveness, the elasticity of an aggregate outcome is
not a single parameter but depends on how outcomes are aggregated. The
geometric elasticity \(\varepsilon_0\) is invariant to \(x\) and
corresponds to the response of the log-aggregate (a typical-unit
notion). This also means that OLS generally estimates the average
response. By contrast, elasticities relevant for aggregates in levels
(including the arithmetic mean) generally vary with \(x\) because
changes in \(x\) tilt the relative importance of units with different
elasticities. The knife-edge condition in
\hyperref[mycor-nowedge]{Corollary~\ref*{mycor-nowedge}} clarifies when
this reweighting is absent, while the Gaussian example illustrates how
the entire family of power-mean elasticities can be characterized in
closed form.

\section{Discussion}\label{sec-discussion}

A natural question that arises is why economists may care to estimate
one power mean elasticity over another. This section motivates why
different economic decision problems depend on different aggregation
rules, and therefore on different elasticities. We give two examples and
appeal to an axiomatisation of a common decision problem. We show that
power mean elasticities are sufficient statistics for this decision
problem. Finally, we note that different power means correspond to
different risk or inequality aversion for the decision maker facing this
problem.

\begin{myexm}[Demand, profits, and welfare]\label{myexm-demand-monopolist}

Let \(p>0\) denote price. Each type \(\omega\) has individual demand \[
D(p;\omega)=\exp(a(\omega))\,p^{\varepsilon(\omega)},
\] so that \(\varepsilon(\omega)\equiv d\log D(p;\omega)/d\log p\) is
the individual price elasticity. Aggregate demand is \[
D(p)\equiv \mathbb{E}[D(p;\omega)].
\] The elasticity of aggregate demand in levels is the arithmetic-mean
elasticity \[
\varepsilon_1(p)
\equiv
\frac{d\log D(p)}{d\log p}
=
\frac{\mathbb{E}\!\left[D(p;\omega)\,\varepsilon(\omega)\right]}{\mathbb{E}[D(p;\omega)]}.
\] By contrast, the elasticity of the geometric mean of demand is \[
\varepsilon_0\equiv \mathbb{E}[\varepsilon(\omega)].
\]

Consider a monopolist with constant marginal cost \(c\) choosing price
\(p\) to maximize expected their profits. They face decision problem

\[
\max_p \ \pi(p) \equiv (p-c)\,D(p).
\] The first-order condition can be written as \[
0=\frac{d\pi(p)}{dp}
=
D(p) + (p-c)D'(p)
\quad\Longleftrightarrow\quad
\frac{p-c}{p}
=
-\frac{1}{\varepsilon_1(p)}.
\] Thus, the monopolist's optimal markup depends on the elasticity of
\textbf{aggregate demand in levels}, \(\varepsilon_1(p)\), not on
\(\varepsilon_0\).

Intuitively, profits depend on total quantity sold. Types with higher
demand receive more weight in determining how total demand responds to
price, and therefore more weight in the elasticity relevant for markups.
It doesn't matter if the monopolist is risk averse, as long as
\(\pi(p)\) is a strictly increasing, differentiable function of
\(D(p)\), they will still wish to estimate the aggregate demand
elasticity.

\end{myexm}

\begin{myexm}\label{myexm-demand-planner}

Now consider a planner who chooses \(p\) to maximize the expected
utility of consumers, taking production costs into account. Suppose each
type has quasi-linear preferences with indirect utility \[
U(\omega;p)=v(\omega)+\log D(p;\omega) - c\,D(p;\omega),
\] where \(v(\omega)\) collects terms independent of \(p\). Then
expected welfare is \[
W(p)\equiv \mathbb{E}[U(\omega;p)].
\] Differentiating, \[
\frac{dW(p)}{d\log p}
=
\frac{d}{d\log p}\mathbb{E}[\log D(p;\omega)]
\;-\;
c\,\frac{d}{d\log p}\mathbb{E}[D(p;\omega)].
\] The first term depends on the elasticity of the \textbf{geometric
mean of demand}, \[
\frac{d}{d\log p}\mathbb{E}[\log D(p;\omega)]
=
\varepsilon_0,
\] while the second depends on the elasticity of the \textbf{arithmetic
mean}, \[
\frac{d}{d\log p}\mathbb{E}[D(p;\omega)]
=
\varepsilon_1(p)\,\mathbb{E}[D(p;\omega)].
\] As a result, the planner's optimal pricing condition combines two
distinct elasticities: one governing changes in average log demand (a
typical-unit object) and one governing changes in total demand (a level
aggregate). Whenever the planner's objective features log terms the
geometric elasticity naturally enters welfare derivatives.

\end{myexm}

Generalising the two previous examples, we provide the following
axiomatisation of a natural decision problem for which power mean
elasticities are sufficient statistics. Consider a statistical decision
problem
\(\mathcal{D} = \langle \Omega, \mathcal{X}, \mathcal{Y}, \succeq \rangle\),
where a Decision Maker (DM) chooses a policy
\(x \in \mathcal{X} \subseteq \mathbb{R}_{+}\), or distribution of
policies \(F_{x }\in\Delta(\mathcal{X})\), to maximize a functional over
a population of heterogeneous units \(\Omega\) with outcomes
\(\mathcal{Y}\).

The DM believes that the units the inputs and outcomes are denominated
in should not matter to the intensity of response for each type.
\hyperref[axm-scale-invariance]{Axiom~\ref*{axm-scale-invariance}}
implies that types respond to percentage changes in inputs with
percentage changes in outcomes, ruling out level-dependent functional
forms.

\begin{axm}[Individual Scale Invariance]\label{axm-scale-invariance}

For every type \(\omega \in \Omega\), the functional relationship
between inputs \(x\) and outcomes \(Y\) is invariant to scaling.
Specifically, for any scaling factor \(\lambda > 0\),

\[Y(\lambda x; \omega) = g(\lambda; \omega) Y(x; \omega)\]

where \(g\) is a unit-specific scalar function.

\end{axm}

Standard assumptions guarantee that a complete, continuous preorder
\(\succsim\) over outcomes admits a numerical representation
\(M: \mathcal{Y}^\Omega \to \mathbb{R}\). We provide four further axioms
that describe how \(M\) should behave. Firstly,
\hyperref[axm-pareto]{Axiom~\ref*{axm-pareto}} states \(M\) should
increase in Pareto improvements. Secondly,
\hyperref[axm-anonymity]{Axiom~\ref*{axm-anonymity}} requires that the
DM does not care about the identities of any of the types. Thirdly,
\hyperref[axm-decomposability]{Axiom~\ref*{axm-decomposability}} is a
separability condition. \(M\) should be separable across types and the
aggregate value of the total population can be computed solely from the
aggregate values of its disjoint subpopulations and their relative
masses. Finally,
\hyperref[axm-homotheticity]{Axiom~\ref*{axm-homotheticity}} requires
the decision maker themselves not care about the units the outcome is
denominated in.

\begin{axm}[Pareto efficiency]\label{axm-pareto}

The aggregator \(M\) is strictly monotone in \(Y(\cdot; \omega)\) for a
set of \(\omega\) with positive measure.

\end{axm}

\begin{axm}[Anonymity]\label{axm-anonymity}

\(M\) is invariant to measure-preserving permutations of \(\omega\).

\end{axm}

\begin{axm}[Decomposability]\label{axm-decomposability}

For any partition \(\Omega=A \cup B\) there exists a function \(\Psi\)
such that \(M(\Omega)=\Psi(M(A),M(B),\mu(A),\mu(B))\) where \(\mu\)
denotes a measure from \((\Omega,\Sigma)\) to \([0,1]\).

\end{axm}

\begin{axm}[Homotheticity]\label{axm-homotheticity}

The ranking of outcomes is invariant to the unit of measurement of the
outcome variable. For any \(k > 0\),

\[M(k \cdot Y) = k \cdot M(Y).\]

\end{axm}

These axioms uniquely identify the structural model and the objective
function.

\begin{myprp}[Representation]\label{myprp-representation}

A decision problem satisfying
\hyperref[axm-scale-invariance]{Axiom~\ref*{axm-scale-invariance}} -
\hyperref[axm-homotheticity]{Axiom~\ref*{axm-homotheticity}} is
isomorphic to maximizing the power mean of a Random Coefficients Model

\[\max_{F_{x}} \ M_\phi(x) \equiv \left( \mathbb{E}_{\omega,x} \left[ \left( \alpha(\omega) x^{\varepsilon(\omega)} \right)^\phi \right] \right)^{1/\phi}\]

for some aggregation parameter \(\phi \in \mathbb{R}\) and unit-specific
parameters \(\alpha(\omega) > 0, \varepsilon(\omega) \in \mathbb{R}\).

\end{myprp}

This representation justifies the potential outcomes random elasticities
model used in the literature, where
\(\log Y(x; \omega) = a(\omega) + \varepsilon(\omega) \log x\), with
\(a(\omega) = \log \alpha(\omega)\). Furthermore, for any decision
problem satisfying our axioms the aggregate outcome level and the
power-mean elasticity are sufficient statistics.

\begin{myprp}[Sufficiency of the Power-Mean Elasticity]\label{myprp-sufficiency}

For any decision problem satisfying
\hyperref[axm-scale-invariance]{Axiom~\ref*{axm-scale-invariance}} -
\hyperref[axm-homotheticity]{Axiom~\ref*{axm-homotheticity}}, the
optimal policy \(F_{x}^{*}\) is determined solely by the aggregate
outcome level \(M_\phi(x)\) and the power-mean elasticity

\[\varepsilon_\phi(x) \equiv \frac{d \log M_\phi(x)}{d \log x}.\]

\end{myprp}

The parameter \(\phi\) captures the curvature of the aggregation. As is
commonly known, this parameter nests two distinct, but incompatible,
economic interpretations. We may interpret the optimization problem
\(\max_x M_\phi(x)\) as equivalent to:

\begin{enumerate}
\def\labelenumi{\arabic{enumi}.}
\tightlist
\item
  a utilitarian planner aggregating constant relative risk aversion
  (CRRA) utility functions with coefficient of relative risk aversion
  \(\gamma = 1-\phi\).
\item
  an Atkinson social planner with inequality aversion parameter
  \(\rho = 1-\phi\) aggregating risk-neutral agents.
\end{enumerate}

To demonstrate, consider a utilitarian social planner who evaluates
policy \(t\) according to the social welfare function

\[ W(t) = \int u(y(t, \omega)) \  dF(\omega) \tag{3}\]

where \(y(t, \omega)\) denotes the outcome realized by individual
\(\omega\) under policy \(t\), and \(F(\omega)\) represents the
distribution of individual characteristics in the population. Suppose
\(u(\cdot)\) takes the form

\[ u(z) = \begin{cases} \dfrac{z^{1-\rho}}{1-\rho} & \text{if } \rho \neq 1 \\ \log(z) & \text{if } \rho = 1 \end{cases} \]

where \(\rho \geq 0\) parameterizes the degree of inequality aversion.
Then

\[
W(t)=\int \frac{y(t,\omega)^{1-\rho}}{1-\rho} dF(\omega).
\]

Alternatively, consider an inequality averse social planner who
evaluates policy \(t\) according to the social welfare function

\[
W(t)=\int \frac{v(y(t,\omega))^{1-\rho}}{1-\rho} dF(\omega).\tag{4}
\]

Equations \((3)\) and \((4)\) are only equal if \(v\) is the identity
function. The choice among power-mean elasticities therefore either
encodes an implicit choice among social welfare functions or an
assumption about the utility of the individuals being aggregated over.
The monopolist's problem depends on the response of \emph{total}
quantity sold and therefore on \(\varepsilon_1(p)\). A planner's welfare
calculus depends on how price changes affect \emph{average log outcomes}
(or other concave aggregates), and therefore on \(\varepsilon_0\) (or,
more generally, \(\varepsilon_\phi(p)\) for \(\phi\neq 1\)). When
elasticities are heterogeneous, these objects differ, so the elasticity
relevant for profit maximization need not coincide with the elasticity
relevant for welfare analysis. Yet both are represented in our
framework.

Different power-mean elasticities may be of interest, even for the same
problem. We characterized the wedge between these elasticities. This
wedge can be identified from observables under assumptions we discuss in
the next section. Once identified, our approach to estimation follows
from \citet{manningLoggedDependentVariable1998}, who estimates these
wedges in a specific setting with binary treatment variables and no
controls.

\section{Identification}\label{sec-identification}

We have shown power mean elasticities are sufficient statistics for a
common class of decision problems. When can we estimate them? This is a
question of identification. We outline the statistical model and discuss
two common identification strategies - randomised controlled trials and
instrumental variables.

An alternative way to write the potential outcome model is \[
\log{y_{i}}(x) = \alpha_{0} + \beta_{0} \log{x} + u_{i}(x)
\] with
\(u_{i}(x) =\alpha_{i} - \alpha_{0} + (\beta_{i} - \beta_{0}) \cdot \log{x}\)
such that \(\mathbb{E}[u_{i}(x)] = 0 \  \forall x\).

Suppose, the \(X\) we observe in the population is randomly distributed
and mean independent of the random coefficients. We observe \(Y\) such
that \(\log Y = \log y(X)\). Then, we can write our population model as
\[
\log Y_{i} = \alpha_{0} + \beta_{0} \log X_{i} + u_{i}, \quad \mathbb{E}[u_{i}|X_{i}] = 0.
\]

Under this model \(\beta_{0}\) is identified, OLS is an unbiased
estimator of it, and \(\beta_{0}=\epsilon_{0}\) - the average elasticity
of the population.

However, if we try to estimate a different elasticity, for example the
causal arithmetic mean elasticity \(\epsilon_{1}\), we will not be able
to do so without the stronger assumption that
\(X \perp \!\! \perp (\alpha_{i},\beta_{i})\). Independence here
requires both no selection on level and no selection on gains. We may
believe this under random assignment or similar assumptions.

Even in a randomized controlled trial with binary treatment,
retransformation bias still warrants attention as

\[\frac{\mathbb{E}[Y(1; \omega) - Y(0; \omega)]}{\mathbb{E}[Y(0; \omega)]} \neq \exp\big(\mathbb{E}[\log Y(1; \omega) - \log Y(0; \omega)]\big) - 1.\]

Under the heterogeneous semi-elasticity model where
\(\log Y(x;\omega) = a(\omega) + \varepsilon(\omega) x\) for
\(x \in \{0,1\}\), OLS on \(\log Y_i\) gives a coefficient estimate that
converges to the average treatment effect in logs,
\(\mathbb{E}[\varepsilon(\omega)] \equiv \varepsilon_0\). Hence,
parameters obtained from log-linear OLS do not directly correspond to
the percentage difference in average outcomes between the treated and
untreated.

However, in the binary treatment case, it is always possible to directly
do OLS in levels, or PPML to get the percentage change in the arithmetic
mean. In fact, the \citet{manningLoggedDependentVariable1998} estimator
(which is equivalent to ours with a binary variable) gives exactly the
same point estimates as PPML with binary treatments (see
Section~\ref{sec-manningppml}).

This is because with a single binary treatment the exponential and
log-linear models are compatible. The log-linear statistical model
decomposes as

\[\log Y_{i} = \alpha_{0} + \varepsilon_{0} x_{i} + u_{i}(x_i), \quad \mathbb{E}[u_{i}(x_i)|x_{i}] = 0\]

where \(\alpha_0 = \mathbb{E}[a(\omega)]\),
\(\varepsilon_0 = \mathbb{E}[\varepsilon(\omega)]\), and
\(u_i(x) = (a(\omega) - \alpha_0) + (\varepsilon(\omega) - \varepsilon_0)x\).

The corresponding exponential model is

\[Y_{i} = \exp(\gamma_{0} + \gamma_{1} x_{i}) \eta_{i}, \quad \mathbb{E}[\eta_{i}|x_{i}] = 1\]

with the parameters mapping exactly as

\[\begin{aligned}
\gamma_{0} &= \alpha_{0} + \log\mathbb{E}[\exp(u_i(0))|x_i=0] \\
\gamma_{1} &= \varepsilon_{0} + \log \mathbb{E}[\exp(u_i(1))|x_i=1] - \log \mathbb{E}[\exp(u_i(0))|x_i=0]\\
\eta_{i} &= \frac{\exp(u_{i}(x_i))}{\mathbb{E}[\exp(u_i(x_i))|x_{i}]}.
\end{aligned}\]

This equivalence does not hold if there is a randomly assigned
continuous treatment \(x\) or other regressors.

When random assignment is violated we must rely on instrumental
variables approaches. However, instrumental variables approaches face
challenges in this setting. Instrumental variables approaches cannot,
under standard assumptions, recover the arithmetic mean elasticity from
log-linear structural models with heterogeneous individual effects. We
show a sequence of impossibility results of increasing generality. We
then show that assuming a triangular first-stage structure restores
identification.

Let each individual have potential outcome
\[ \log Y_i(x) = a_i + \varepsilon_i \log x \] where \(a_i\) is an
individual intercept and \(\varepsilon_i\) is an individual
semi-elasticity, with joint distribution
\((a_i, \varepsilon_i) \sim F\). Write
\(\beta_0 = \mathbb{E}[\varepsilon_i]\) and
\(\alpha_0 = \mathbb{E}[a_i]\). The geometric mean semi-elasticity is
\(\mathbb{E}[\varepsilon_i] = \beta_0\), constant in \(x\). The
arithmetic mean semi-elasticity is
\[ \theta(x) = \frac{\partial \log \mathbb{E}[Y_i(x)]}{\partial x} = \frac{\mathbb{E}[\varepsilon_i  e^{a_i + \varepsilon_i \log x}]}{\mathbb{E}[e^{a_i + \varepsilon_i \log x}]}, \]
a level-weighted average of individual semi-elasticities that varies
with \(x\).

We observe \((Y_i, X_i, Z_i)\) where
\(\log Y_i = a_i + \varepsilon_i X_i\) and \(Z_i\) is an instrument for
\(X_i\). The statistical model decomposes as
\[ \log Y_i = \alpha_0 + \beta_0 \log X_i + u_i(X_i), \qquad u_i(x) = (a_i - \alpha_0) + (\varepsilon_i - \beta_0) \log x, \]
with \(\mathbb{E}[u_i(x)] = 0\) for all \(x\).

First, let's consider standard IV assumptions. Let \(Z\) satisfy mean
independence
\(\mathbb{E}[\varepsilon_i \mid Z] = \mathbb{E}[\varepsilon_i]\) and
\(\mathbb{E}[a_i \mid Z] = \mathbb{E}[a_i]\), together with relevance
\(\text{Cov}(X,Z) \neq 0\). Under these conditions, IV consistently
estimates \(\beta_0 = \mathbb{E}[\varepsilon_i]\), the geometric mean
elasticity.

\begin{myprp}[Mean independence does not identify \(\theta(x)\)]\label{myprp-iv-impossibility}

As is known, under the given conditions, the arithmetic mean elasticity
\(\theta(x)\) is not identified by IV, even when IV consistently
estimates \(\mathbb{E}[\varepsilon_i]\).

\end{myprp}

IV assumptions constrain only first moments of \((a_i, \varepsilon_i)\)
conditional on \(Z\). The arithmetic mean elasticity depends on
\(\mathbb{E}[\varepsilon_i  e^{a_i + \varepsilon_i \log x}]\), which
involves the entire joint distribution of \((a_i, \varepsilon_i)\). Our
proof provides two observationally equivalent data-generating processes
that share the same \(\beta_0\) but yield different values of
\(\theta(x)\). In one, the observed heteroskedasticity of residuals
arises from selection (high-variance individuals choose certain
treatment values). In the other, it arises from genuine heterogeneity in
\(\varepsilon_i\) across the population. Since IV restricts only
conditional means, it cannot distinguish these mechanisms.

Strengthening IV to a non-parametric specification, identifying \(f(x)\)
via completeness, does not resolve the problem since \(\theta(x)\)
depends on the joint distribution of \((a_{i},\epsilon_{i})\) not just
conditional means.

One might hope that continuous instruments with full independence might
provide identification . The following result shows that they cannot.
For any model satisfying mild regularity conditions, there exists an
observationally equivalent model with a different arithmetic mean
semi-elasticity.

\begin{myprp}[Full independence with continuous instruments does not identify
\(\theta(x)\)]\label{myprp-iv-exact-nonid}

Let \((X, Z)\) be continuously distributed and assume full independence
\(Z \perp\!\!\perp (a_i, \varepsilon_i)\). Let \((f_0, k_0)\) be any
model with \(f_0\) continuous and strictly positive, \(k_0\) continuous
and strictly positive on \(\text{supp}(X)\) with
\(\inf_z k_0(x \mid a, \varepsilon, z) > 0\) for each
\((a, \varepsilon, x)\), and
\(\mathbb{E}[e^{a + \varepsilon \log x}] < \infty\) in a neighbourhood
of some \(x_0 \neq 0\). There exists an observationally equivalent model
\((f_1, k_1)\) with \(\theta_1(x_0) \neq \theta_0(x_0)\).

\end{myprp}

The proof exploits the geometry of the identification problem. Write
\(h(a, \varepsilon, x, z) = k(x \mid a, \varepsilon, z) f(a, \varepsilon)\)
for the joint density of \((a_i, \varepsilon_i, X_i)\) conditional on
\(Z = z\). The structural equation \(\log Y = a + \varepsilon \log x\)
implies that for each fixed \((x, z)\), the observable density
\(f(\ell, x \mid z)\) is the Radon projection of
\(h(\cdot, \cdot, x, z)\) along lines
\({(a, \varepsilon) : a + \varepsilon \log x = \ell}\). Without
triangularity, \(h(\cdot, \cdot, x, z)\) is a different unknown function
for each \((x, z)\). We observe one projection of each, rather than many
projections of one. The construction produces a compactly supported
perturbation of \(f\) that lies in the null space of every relevant
Radon projection, thereby changing the joint distribution of
\((a_i, \varepsilon_i)\) without affecting any observable.

With full independence of discrete instruments for discrete treatments,
an arithmetic mean elasticity for compliers is identified via standard
LATE arguments as per \citet{imbensIdentificationEstimationLocal1994}.

The preceding results establish that IV methods cannot identify
\(\theta(x)\) without further structural restrictions. A triangular
first-stage assumption resolves this by collapsing the selection
mechanism to a single dimension.

\begin{myprp}[Control function assumptions identify \(\theta(x)\)]\label{myprp-control-function-identification}

Assume a triangular first stage \(X = g(Z, V)\) with \(g\) strictly
monotone in \(V\), \(V \perp\!\!\perp Z\), and
\((a_i, \varepsilon_i) \perp\!\!\perp Z \mid V\). Suppose that for each
\(v\) in the support of \(V\), the map \(z \mapsto g(z, v)\) generates a
set of treatment values \({g(z,v) : z \in \text{supp}(Z)}\) with
nonempty interior. Then the conditional distribution
\(f(a, \varepsilon \mid V = v)\) is identified for each \(v\), and the
arithmetic mean elasticity \(\theta(x)\) is identified.

\end{myprp}

Triangularity and the sufficient support condition ensures that the
Radon transform, restricted to the angles generated, is injective on
\(L^2(\mathbb{R}^2)\), so \(f(a, \varepsilon \mid V = v)\) is fully
recovered. Integrating over \(f_V\) then gives
\[ \mathbb{E}[e^{a + \varepsilon \log x}] = \int \mathbb{E}[e^{a + \varepsilon \log x} \mid V = v]\ f_V(v)\ dv \]
where each conditional expectation is identified.

\section{Debiased Estimation}\label{sec-debiased-estimators}

Regardless of whether causal identification holds, we may wish to
estimate the arithmetic mean elasticity of a log-log model. As a
semi-parametric estimator the estimator we provide in this section will
converge to the average arithmetic mean elasticity under regularity
assumptions. Even when economists run log-log (log-linear) projection
models, they can use our estimator to obtain the average arithmetic
(semi-)elasticity of the observed variables.

We present estimators that converge under general conditions to the
average arithmetic mean elasticity. In the first part, we provide an
estimator which applies a correction to the log-log model directly which
will give causal estimates under full independence of the regressor and
random coefficients. In the second part, we provide an instrumental
variables estimator which uses a control function approach. We give
conditions for their convergence and asymptotic distribution.

As we must estimate high-dimensional functions we turn to machine
learning methods to avoid the curse of dimensionality.
Neyman-orthogonalised debiased machine learning methods
\citep{chernozhukovLocallyRobustSemiparametric2022, chernozhukovDoubleDebiasedMachine2018}
permit root-\(n\) convergence even with neural networks.

We estimate the average arithmetic mean elasticity (semi-elasticity)
starting from a log-log (log-linear) model. We cannot provide faster
converging estimators for the non-parametric power mean sufficient
statistic in our decision problem than other non-parametric estimators -
though our semi-parametric estimator for the \textbf{average} converges
faster than the naive ``plug-in'' estimator
\citep{chernozhukovLocallyRobustSemiparametric2022}. However, we note
that if the decision problem assumptions hold then estimating a
non-parametric model on the residuals of the log-log model will have
strictly lower asymptotic variance to second order than a general
non-parametric approach without the isoelastic heterogeneity assumption.
We don't yet provide asymptotics for non-parametric estimation.

Our estimators develop on
\citet{aiSemiparametricDerivativeEstimator2008} who use splines to
estimate the conditional mean of the outcome variable and its derivative
in a non-linear model under transformation.

Researchers may want to understand the arithmetic percentage change
rather than the arithmetic elasticity when regressors are discrete. We
have developed this estimator but do not provide it here yet. Some of
our results in Section~\ref{sec-empirical} use this estimator.

\subsection{Ordinary estimator}\label{ordinary-estimator}

Let \(W=(Y,X,Z)\), where \(X\in\mathbb R^{d_x}\) are the variables of
interest (all derivatives are taken with respect to \(x\)) and \(Z\) are
controls. Assume the outcome satisfies \[
\log Y = \beta_0^\top X + \gamma_0^\top Z + u,
\qquad
\mathbb E[u\mid X,Z]=0.
\] This is isomorphic to the random coefficients model when observed X
is independent from \(\alpha(\omega)\) and \(\varepsilon(\omega)\).

Define the primitive nuisance object \[
m_0(x,z) := \mathbb E[e^u\mid X=x,Z=z],
\qquad
0<\underline m \le m_0(x,z)\le \overline m<\infty.
\] Then \[
\mathbb E[Y\mid X=x,Z=z]
=
e^{\beta_0^\top x+\gamma_0^\top z}\,m_0(x,z),
\] and \[
\nabla_x \log \mathbb E[Y\mid X,Z]
=
\beta_0 + \nabla_x \log m_0(X,Z).
\]

We target the average semi-elasticity of the conditional mean with
respect to \(X\): \[
\theta_0
:=
\mathbb E\!\left[\nabla_x \log \mathbb E[Y\mid X,Z]\right]
=
\mathbb E\!\left[\beta_0 + \frac{\nabla_x m_0(X,Z)}{m_0(X,Z)}\right].
\] Elasticities with respect to \(Z\) are not targeted. Note that if X
was log K for some variable K, this would be the average arithmetic mean
elasticity with respect to K.

Let \(f_0(x\mid z)\) denote the conditional density of \(X\) given
\(Z=z\), differentiable in \(x\) with \(f_0(x\mid z)>0\). Define \[
\alpha_0(x,z)
:=
-\frac{\nabla_x f_0(x\mid z)}{f_0(x\mid z)\,m_0(x,z)}.
\]

Define the residual \[
\delta := e^u - m_0(X,Z),
\qquad
\mathbb E[\delta\mid X,Z]=0.
\]

Note that \[
\alpha_0(X,Z)\,\delta
= -
\frac{\nabla_x f_0(X\mid Z)}{f_0(X\mid Z)}
\left(\frac{e^u}{m_0(X,Z)}-1\right),
\]

which is the standard locally robust correction term.

Let \(m_0'(x,z):=\nabla_x m_0(x,z)\).

Define the score \[
\phi(W;\theta, \beta, \gamma, m, f)
:=
\beta + \frac{m'(X,Z)}{m(X,Z)} - \theta
\;+\;
\alpha(X,Z)\,\big(e^{u(\beta,\gamma)} - m(X,Z)\big),
\] where \[
u(\beta,\gamma) := \log Y - \beta^\top X - \gamma^\top Z,
\qquad
\alpha(x,z) := \frac{\nabla_x f(x\mid z)}{f(x\mid z)\,m(x,z)}.
\]

At the true parameter values \((\theta_0,\beta_0,\gamma_0,m_0,f_0)\), \[
\mathbb E[\phi(W;\theta_0)]=0.
\]

\begin{theorem}[Neyman
Orthogonality]\protect\hypertarget{thm-neymanorthogonality}{}\label{thm-neymanorthogonality}

\(\phi\) is Neyman-orthogonal with respect to the nuisance functions
\((m,f)\) and the finite-dimensional parameters \((\beta,\gamma)\).

\end{theorem}

A formal proof is in Section~\ref{sec-noproof}.

\begin{enumerate}
\def\labelenumi{\arabic{enumi}.}
\item
  Estimate \((\beta_0,\gamma_0)\) and form residuals \[
  \hat u = \log Y - \hat\beta^\top X - \hat\gamma^\top Z.
  \]
\item
  Estimate \[
  m_0(x,z)=\mathbb E[e^u\mid X=x,Z=z]
  \] by regressing \(e^{\hat u}\) on \((X,Z)\) using a flexible learner
  (e.g.~a neural network), yielding \(\hat m(x,z)\) and \[
  \hat m'(x,z)=\nabla_x \hat m(x,z)
  \] (e.g.~via automatic differentiation).
\item
  Estimate the conditional density \(f_0(x\mid z)\) and form \[
  \hat\alpha(x,z)= - \frac{\nabla_x \hat f(x\mid z)}{\hat f(x\mid z)\,\hat m(x,z)}.
  \]
\item
  Use sample splitting and cross-fitting. Define \(\hat\theta\) as the
  solution to \[
  \frac1n\sum_{i=1}^n
  \phi\!\left(W_i;\theta,\hat\beta,\hat\gamma,\hat m,\hat f\right)
  =0,
  \] with each observation evaluated using nuisance estimates trained on
  a fold that excludes it. In our estimation, we also estimate beta and
  gamma on the fold excluding the observation.
\end{enumerate}

We call our estimator the doubly robust estimator of the arithmetic mean
(DREAM).

\begingroup

\renewcommand{\thetheorem}{A1}

\begin{asm}[Sampling]\label{asm-ordinary-sampling}

\(\{(Y_i,X_i,Z_i)\}_{i=1}^n\) are i.i.d.

\end{asm}

\addtocounter{theorem}{-1}

\endgroup

\hyperref[asm-ordinary-sampling]{Assumption~\ref*{asm-ordinary-sampling}}
can be weakened, and the estimator still converges, for example in fixed
effects models.

\begingroup

\renewcommand{\thetheorem}{A2}

\begin{asm}[Regularity]\label{asm-ordinary-regularity}

\(m_0(x,z)\) is bounded away from zero and infinity; \(f_0(x\mid z)\) is
continuously differentiable in \(x\) and strictly positive. The
integration-by-parts condition in the appendix A of holds.

\end{asm}

\addtocounter{theorem}{-1}

\endgroup

The first condition requires all moments of \(u|X=x\) to be finite.
Positivity is ensured as \(e^u > 0 \quad \forall u.\)

\begingroup

\renewcommand{\thetheorem}{A3}

\begin{asm}[Moments and identification]\label{asm-ordinary-moments}

\(\mathbb E[\phi(W;\theta_0)^2]<\infty\) and \(\theta_0\) is uniquely
identified by \[
\mathbb E[\phi(W;\theta)]=0.
\]

\end{asm}

\addtocounter{theorem}{-1}

\endgroup

\hyperref[asm-ordinary-moments]{Assumption~\ref*{asm-ordinary-moments}}
is a standard GMM condition, which ensures that \(\theta_{0}\) is
identifiable and its estimators have finite variance.

\begingroup

\renewcommand{\thetheorem}{A4}

\begin{asm}[Nuisance rates with cross-fitting]\label{asm-ordinary-nuisance-rates}

The nuisance estimators satisfy product-rate conditions sufficient for
locally robust Z-estimation, for example, \[
\left\|\nabla_x\log \hat m - \nabla_x\log m_0\right\|_{L^2}=o_p(1),
\qquad
\|\hat m-m_0\|_{L^2}\cdot\|\hat\alpha-\alpha_0\|_{L^2}
=o_p(n^{-1/2}),
\] and analogous conditions for the components entering \(\hat u\)
through \((\hat\beta,\hat\gamma)\).

\end{asm}

\addtocounter{theorem}{-1}

\endgroup

For a sufficiently smooth m and density, both these conditions are
satisfied.

Let \[
\phi_0(W) := \phi(W;\theta_0,\beta_0,\gamma_0,m_0,f_0),
\qquad
V := \mathbb E[\phi_0(W)^2].
\]

Since \[
\partial_\theta \mathbb E[\phi(W;\theta)]\big|_{\theta=\theta_0} = -1,
\] the influence function equals \(\phi_0(W)\).

Under assumptions
\hyperref[asm-ordinary-sampling]{Assumption~\ref*{asm-ordinary-sampling}}
--
\hyperref[asm-ordinary-nuisance-rates]{Assumption~\ref*{asm-ordinary-nuisance-rates}}
and the Neyman orthogonality of the score established in Appendix A,
Theorem 3.1 of \citet{chernozhukovDoubleDebiasedMachine2018} implies

\begin{theorem}[Asymptotic
Normality]\protect\hypertarget{thm-asympnorm}{}\label{thm-asympnorm}

Under
\hyperref[asm-ordinary-sampling]{Assumption~\ref*{asm-ordinary-sampling}}
--
\hyperref[asm-ordinary-nuisance-rates]{Assumption~\ref*{asm-ordinary-nuisance-rates}}
and using cross-fitting, \[
\sqrt n(\hat\theta-\theta_0)
=
\frac{1}{\sqrt n}\sum_{i=1}^n \phi_0(W_i) + o_p(1)
\;\Rightarrow\;
N(0,V).\]

\end{theorem}

A consistent variance estimator is \[
\hat V := \frac1n\sum_{i=1}^n \hat\phi_i^2,
\qquad
\hat\phi_i :=
\phi(W_i;\hat\theta,\hat\beta,\hat\gamma,\hat m_{-k(i)},\hat f_{-k(i)}),\]
and confidence intervals can be constructed as \[
\hat\theta \pm z_{1-\alpha/2}\sqrt{\hat V/n}. \]

\begin{corollary}[]\protect\hypertarget{cor-consistency}{}\label{cor-consistency}

Under
\hyperref[asm-ordinary-sampling]{Assumption~\ref*{asm-ordinary-sampling}}
--
\hyperref[asm-ordinary-nuisance-rates]{Assumption~\ref*{asm-ordinary-nuisance-rates}},
\(\hat{\theta}\) is consistent for \(\theta_{0}\)

\end{corollary}

Neural networks can be used to estimate
\(m_0(x,z)=\mathbb E[e^u\mid x,z]\) and its derivative
\(\nabla_x m_0(x,z)\) via automatic differentiation. Sieve-type
approximation and rate results are given in
\citet{farrellDeepNeuralNetworks2021} and
\citet{schmidt-hieberNonparametricRegressionUsing2020}. In our
implementation, we also use neural networks to estimate
\(\frac{\nabla_x \hat f(x\mid z)}{\hat f(x\mid z)\,}\) using score
matching
\citep[\citet{hanNeuralNetworkbasedScore2024}]{hyvarinenEstimationNonnormalizedStatistical2005}

\subsection{Instrumental variables
estimator}\label{sec-iv-debiased-estimators}

Let \(W = (Y, X, Z)\), where \(X \in \mathbb{R}^{d_x}\) is endogenous
and \(Z \in \mathbb{R}^{d_z}\) are instruments. Assume a triangular
system \[ X = g_0(Z) + V, \qquad Z \perp V, \]
\[ \log Y = \beta_0^\top X + \rho_0^\top V + \epsilon, \qquad \mathbb{E}[\epsilon \mid X, V] = 0. \]
The first-stage residual \(V := X - g_0(Z)\) captures the endogeneity in
\(X\). Conditional on \(V\), variation in \(X\) is driven solely by the
instrument \(Z\) and is therefore exogenous.

Define the primitive nuisance objects
\[ m_0(x, v) := \mathbb{E}[e^u \mid X = x, V = v], \qquad u := \log Y - \beta_0^\top X - \rho_0^\top V, \]
\[ \mu_0(x) := \mathbb{E}_V[m_0(x, V)] = \int m_0(x, v) f_V(v) , dv. \]
The function \(\mu_0(x)\) is the average structural function (ASF),
representing the expected outcome under an intervention setting
\(X = x\) while integrating over the marginal distribution of \(V\).

We target the average semi-elasticity of the ASF with respect to \(X\),
\[ \theta_0 := \mathbb{E}\left[\nabla_x \log \mu_0(X)\right] = \mathbb{E}\left[\beta_0 + \frac{\nabla_x \mu_0(X)}{\mu_0(X)}\right]. \]

Define the density ratio
\[ \omega_0(x, v) := \frac{f_V(v)}{f_{V \mid X}(v \mid x)}, \] and the
marginal score \[ S_X(x) := \nabla_x \log f_X(x). \]

Let \(\mu_0'(x) := \nabla_x \mu_0(x)\). The score is given by
\[ \phi(W; \theta, \beta, \rho, m, g, \omega, S_X, \lambda) := \beta + \frac{\mu'(X)}{\mu(X)} - \theta + \phi^m(W) + \phi^g(W), \]
where the correction for the outcome nuisance \(m\) is
\[ \phi^m(W) := -\frac{\omega(X, V) S_X(X)}{\mu(X)} \left( e^{u(\beta, \rho)} - m(X, V) \right), \]
with \(u(\beta, \rho) := \log Y - \beta^\top X - \rho^\top V\), and the
correction for the first-stage nuisance \(g\) is
\[ \phi^g(W) := -\lambda(Z)^\top (X - g(Z)). \] The function
\(\lambda_0(z)\) is the Riesz representer for the pathwise derivative
with respect to \(g\). Rather than deriving its closed form, we estimate
it via automatic debiasing
\citep{chernozhukovAutomaticDebiasedMachine2022}.

At the true parameter values, \(\mathbb{E}[\phi(W; \theta_0)] = 0\).

\begin{theorem}[Neyman Orthogonality
(IV)]\protect\hypertarget{thm-iv-neymanorthogonality}{}\label{thm-iv-neymanorthogonality}

\(\phi\) is Neyman-orthogonal with respect to the nuisance functions
\((m, g, \omega, S_X, \lambda)\) and the finite-dimensional parameters
\((\beta, \rho)\).

\end{theorem}

A formal proof is in Section~\ref{sec-iv-proof}.

Estimate \(g_0(z) = \mathbb{E}[X \mid Z = z]\) using a flexible learner,
yielding \(\hat{g}(z)\), and form residuals
\[ \hat{V} = X - \hat{g}(Z). \]

Estimate \((\beta_0, \rho_0)\) from the regression
\(\log Y = \beta^\top X + \rho^\top \hat{V} + \epsilon\) and form
residuals
\(\hat{u} = \log Y - \hat{\beta}^\top X - \hat{\rho}^\top \hat{V}\).
Estimate \(m_0(x, v) = \mathbb{E}[e^u \mid X = x, V = v]\) by regressing
\(e^{\hat{u}}\) on \((X, \hat{V})\) using a neural network, yielding
\(\hat{m}(x, v)\).

Compute the marginal integration
\[ \hat{\mu}(x) = \frac{1}{n} \sum_{j=1}^n \hat{m}(x, \hat{V}_j), \qquad \nabla_x \hat{\mu}(x) = \frac{1}{n} \sum_{j=1}^n \nabla_x \hat{m}(x, \hat{V}_j), \]
where \(\nabla_x \hat{m}\) is obtained via automatic differentiation.

Estimate \(\omega_0(x, v) = f_V(v) / f_{V \mid X}(v \mid x)\) using
density ratio estimation or by separately estimating \(f_V\) and
\(f_{V \mid X}\). Estimate \(S_X(x) = \nabla_x \log f_X(x)\) via score
matching. Estimate \(\lambda_0(z)\) using the automatic debiasing
procedure of \citet{chernozhukovAutomaticDebiasedMachine2022}, which
learns the Riesz representer by minimizing the expected squared norm of
the correction term subject to the moment constraint.

Use sample splitting with \(K\) folds. Define \(\hat{\theta}\) as the
solution to
\[ \frac{1}{n} \sum_{i=1}^n \phi\left(W_i; \theta, \hat{\eta}_{-k(i)}\right) = 0, \]
where \(\hat{\eta}_{-k(i)}\) denotes nuisance estimates trained on folds
excluding observation \(i\).

\begingroup

\renewcommand{\thetheorem}{B1}

\begin{asm}[Sampling]\label{asm-iv-sampling}

\({(Y_i, X_i, Z_i)}_{i=1}^n\) are i.i.d.

\end{asm}

\addtocounter{theorem}{-1}

\endgroup

\begingroup

\renewcommand{\thetheorem}{B2}

\begin{asm}[First stage]\label{asm-iv-first-stage}

\(X = g_0(Z) + V\) with \(Z \perp V\), and \(g_0\) is identified.

\end{asm}

\addtocounter{theorem}{-1}

\endgroup

\begingroup

\renewcommand{\thetheorem}{B3}

\begin{asm}[Regularity]\label{asm-iv-regularity}

\(m_0(x, v)\) is bounded away from zero and infinity; \(f_V(v) > 0\) and
\(f_{V \mid X}(v \mid x) > 0\) on the support; the density functions are
continuously differentiable with bounded derivatives.

\end{asm}

\addtocounter{theorem}{-1}

\endgroup

\begingroup

\renewcommand{\thetheorem}{B4}

\begin{asm}[Moments and identification]\label{asm-iv-moments}

\(\mathbb{E}[\phi(W; \theta_0)^2] < \infty\) and \(\theta_0\) is
uniquely identified by \(\mathbb{E}[\phi(W; \theta)] = 0\).

\end{asm}

\addtocounter{theorem}{-1}

\endgroup

\begingroup

\renewcommand{\thetheorem}{B5}

\begin{asm}[Nuisance rates with cross-fitting]\label{asm-iv-nuisance-rates}

The nuisance estimators satisfy product-rate conditions
\[ |\hat{m} - m_0|_{L^2} \cdot |\hat{\alpha} - \alpha_0|_{L^2} = o_p(n^{-1/2}), \]
\[ |\hat{g} - g_0|_{L^2} \cdot |\hat{\lambda} - \lambda_0|_{L^2} = o_p(n^{-1/2}), \]
where \(\alpha_0(x, v) := -\omega_0(x, v) S_X(x) / \mu_0(x)\) is the
Riesz representer for \(m\).

\end{asm}

\addtocounter{theorem}{-1}

\endgroup

Let
\[ \phi_0(W) := \phi(W; \theta_0, \eta_0), \qquad V := \mathbb{E}[\phi_0(W)^2]. \]

\begin{theorem}[Asymptotic Normality
(IV)]\protect\hypertarget{thm-iv-asympnorm}{}\label{thm-iv-asympnorm}

Under \hyperref[asm-iv-sampling]{Assumption~\ref*{asm-iv-sampling}} --
\hyperref[asm-iv-nuisance-rates]{Assumption~\ref*{asm-iv-nuisance-rates}} and using
cross-fitting,
\[ \sqrt{n}(\hat{\theta} - \theta_0) = \frac{1}{\sqrt{n}} \sum_{i=1}^n \phi_0(W_i) + o_p(1) ;\Rightarrow; N(0, V). \]

\end{theorem}

A consistent variance estimator is
\[ \hat{V} := \frac{1}{n} \sum_{i=1}^n \hat{\phi}_i^2, \qquad \hat{\phi}_i := \phi(W_i; \hat{\theta}, \hat{\eta}_{-k(i)}), \]
and confidence intervals can be constructed as
\(\hat{\theta} \pm z_{1-\alpha/2} \sqrt{\hat{V}/n}\).

\begin{corollary}[]\protect\hypertarget{cor-iv-consistency}{}\label{cor-iv-consistency}

Under \hyperref[asm-iv-sampling]{Assumption~\ref*{asm-iv-sampling}} --
\hyperref[asm-iv-nuisance-rates]{Assumption~\ref*{asm-iv-nuisance-rates}},
\(\hat{\theta}\) is consistent for \(\theta_0\).

\end{corollary}

\section{Empirical evidence}\label{sec-empirical}

In this section we present empirical evidence from a re-estimation
exercise. We conducted a census of all articles published in the ``Top
5'' economics journals in the year 2020 (the top 5 journals being
Econometrica, Quarterly Journal of Economics, Journal of Political
Economy, American Economic Review, Review of Economic Studies). Our
census revealed that 36.6\% of the 416 articles published in these 5
journals in 2020 contained regressions with log-dependent variables. Out
of the papers that employed log-DV regressions, approximately 70\% of
them interpreted their coefficients as percentage changes of the
arithmetic mean. We also conducted a census of articles from the first
two issues published in 2020 of a selection top field journals. From
this census of 113 papers, we found that 53\% of papers ran
log-dependent variable regressions.
Figure~\ref{fig-percentage-of-papers} presents the results from this
census.

\begin{figure}

\centering{

\includegraphics[width=1\linewidth,height=\textheight,keepaspectratio]{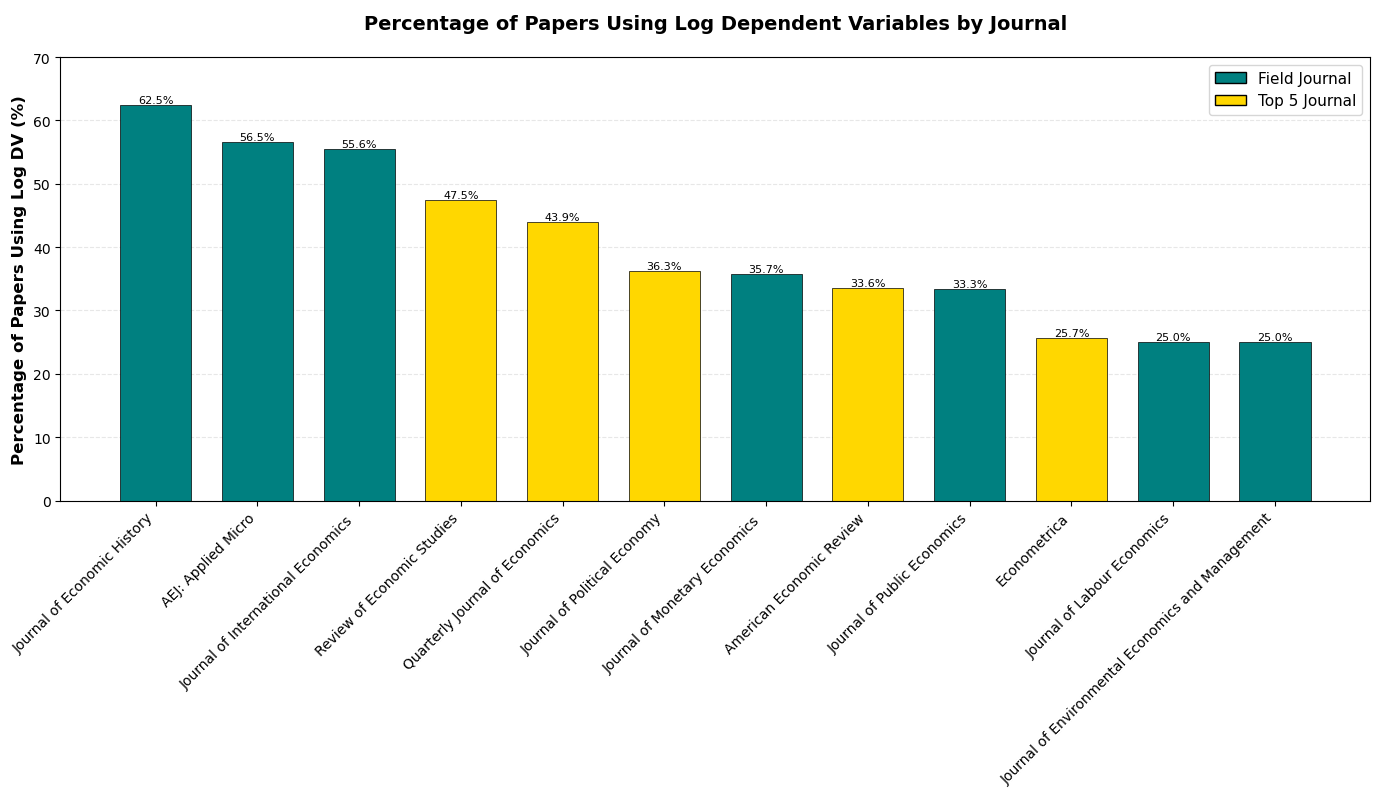}

}

\caption{\label{fig-percentage-of-papers}}

\end{figure}%

After identifying articles published in the Top 5 journals that used log
dependent variable regressions we collected the replication package and
associated datasets for these papers. The inclusion criteria for
replicating were as follows: (0) in our list of top 5 journal articles
published in 2020 with a log-DV regression, (1) publicly available or
easily available data, (2) replication package exists, (3) not
structural estimations. After the data was collected and results of
log-linear or log-log regressions were correctly re-estimated, we
re-estimated results that were not estimated using an instrumental
variable identification strategy.

With these exclusion criteria for non-IV results, we were left with 60
results from 29 papers. For each result we estimated the OLS coefficient
from the log-linear OLS model, the PPML coefficient from the exponential
model, and the doubly robust elasticity of the arithmetic mean (DREAM)
estimator which is consistent for the average semi-elasticity of
\(\mathbb{E}[y_{i}| x_{i}]\). Our re-estimation exercise shows 3 facts:
(1) PPML and OLS estimates are statistically significantly different
from each other in 27\% of the results . If there was full independence
between \(x_i\) and the error term \(v_i\) in the log-linear model, then
the coefficient \(\beta\) would equal the exponential model coefficient
\(\gamma\), because the arithmetic and geometric elasticities would be
equal. (2) OLS and DREAM report statistically significantly different
estimates 32\% of the time. Of the 19 results that are statistically
significantly different from the OLS results, the median absolute
percent difference was 109\%, and the 25th percentile was 67\%. (3) PPML
and DREAM report statistically significantly different estimates 30\% of
the time. Of the 18 statistically significantly different results, the
median absolute percent difference was 75\%, and the 25th percentile was
54\%.\footnote{To appropriately compare elasticity or semi-elasticity
  estimates, we need to treat discrete treatment variables of interest
  differently from continuous ones. Specifically for continuous
  treatment variables, we can run a direct test of the OLS coefficient
  against the doubly-robust semi-elasticity \(\psi\), in this case the
  null is \(H_0: \beta_{log-lin}=\psi\) . For discrete treatment
  variables, the appropriate semi-elasticity comparison is
  \(H_0:e^{\beta_{log-lin}}-1=\psi\) . This holds for elasticities too,
  and works the same for exponential coefficients estimated with PPML.}

Figure~\ref{fig-retransformation-bias-ppml} shows the results for the
difference between OLS and PPML. If there is full independence between
the regressors and the error terms then the OLS coefficient and the PPML
coefficient will converge to the same constant. We recognize that this
exercise of comparing OLS and PPML coefficients is not a definitive test
for independence or reliability of one model over another. We do this to
demonstrate that OLS and PPML coefficients are statistically
significantly different in many settings, and therefore we encourage
researchers to consider that heterogeneity may be a reason for why we
observe significant differences between these two estimates. Here we may
be comparing results from two mutually exclusive models. Therefore, we
encourage readers to only take this as an illustrative example of the
differences of results coming from two estimation strategies of
corresponding (but often conflicting) regression models.

\begin{figure}

\centering{

\includegraphics[width=1\linewidth,height=\textheight,keepaspectratio]{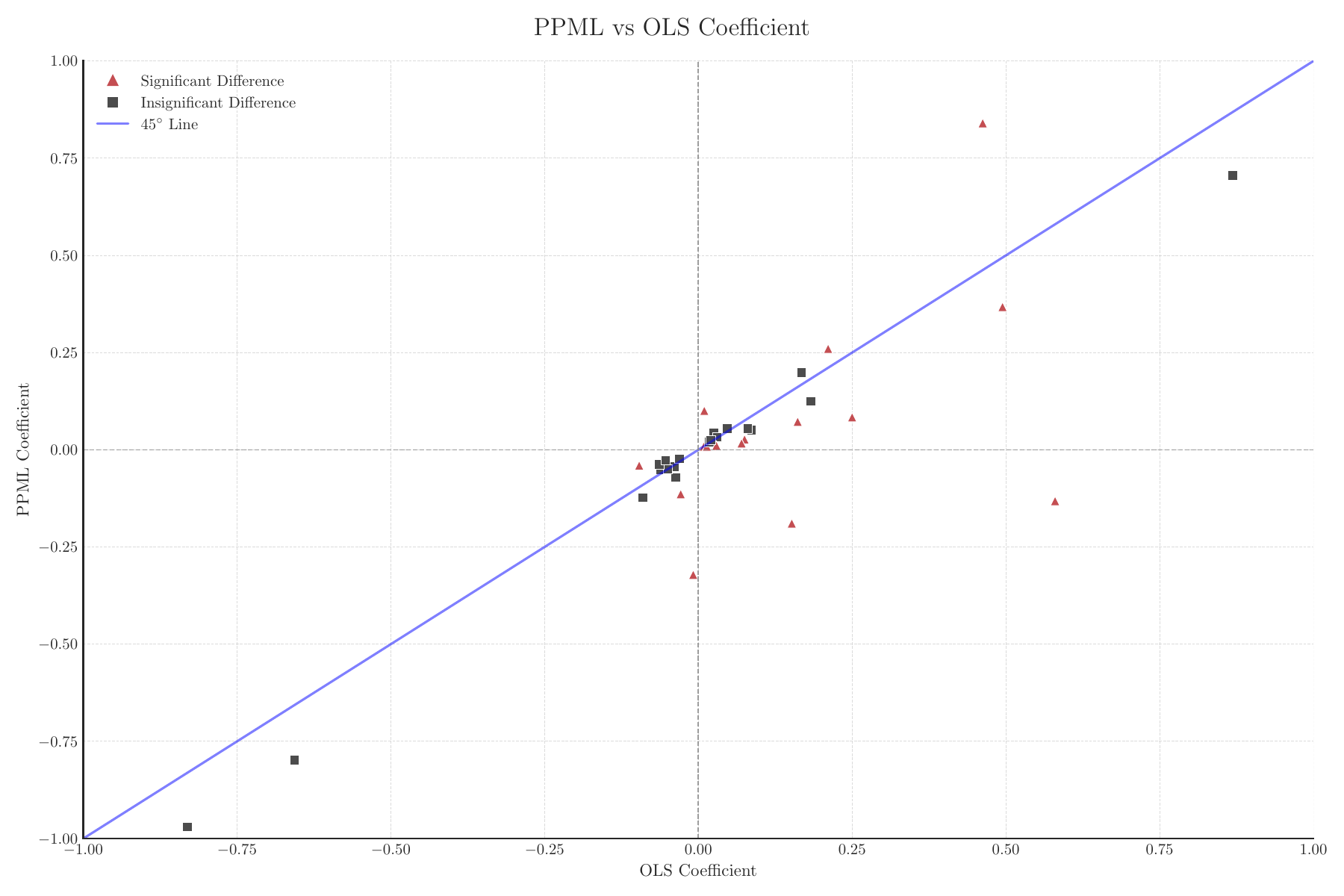}

}

\caption{\label{fig-retransformation-bias-ppml}}

\end{figure}%

Differences in elasticity estimates obtained from PPML and OLS
estimation strategies motivate us to consider using our
specification-robust Neyman Orthogonal estimator. We first compare
elasticities obtained from OLS estimation with our semi-elasticity
estimates from DREAM. Figure Figure~\ref{fig-elast-vs-ols-2d} presents
these results. Here, the OLS elasticity estimate is plotted on the
x-axis, and the DREAM elasticity estimate is plotted on the y-axis. The
parity line is presented for reference. The red triangles represent
estimates where OLS is statistically significantly different from DREAM,
and the black squares represent those where the difference is not
statistically significant. Note, any dot in the second or fourth
quadrant represents a sign reversal -meaning a positive semi-elasticity
was estimated with OLS, and a negative one was estimated with DREAM, or
vice versa. Of the 50 estimations, 18 of the estimates obtained using
OLS are statistically significantly different from those obtained using
DREAM. Of these 18 that are statistically significantly different, the
median absolute bias is 96.4\% of the OLS coefficient. If we restrict
results to those where the absolute value of the average elasticity
estimate is less than 5, then we have 42 total results, 16 of which are
statistically significantly different from the OLS estimator. Of these
16 results, the median absolute bias is 64.97\% of the OLS coefficient.
Our preferred bias estimate removes results with absolute DREAM
estimates greater than 5, as these extreme values likely reflect
instability in the machine learning nuisance parameter estimation.

\begin{figure}

\centering{

\includegraphics[width=1\linewidth,height=\textheight,keepaspectratio]{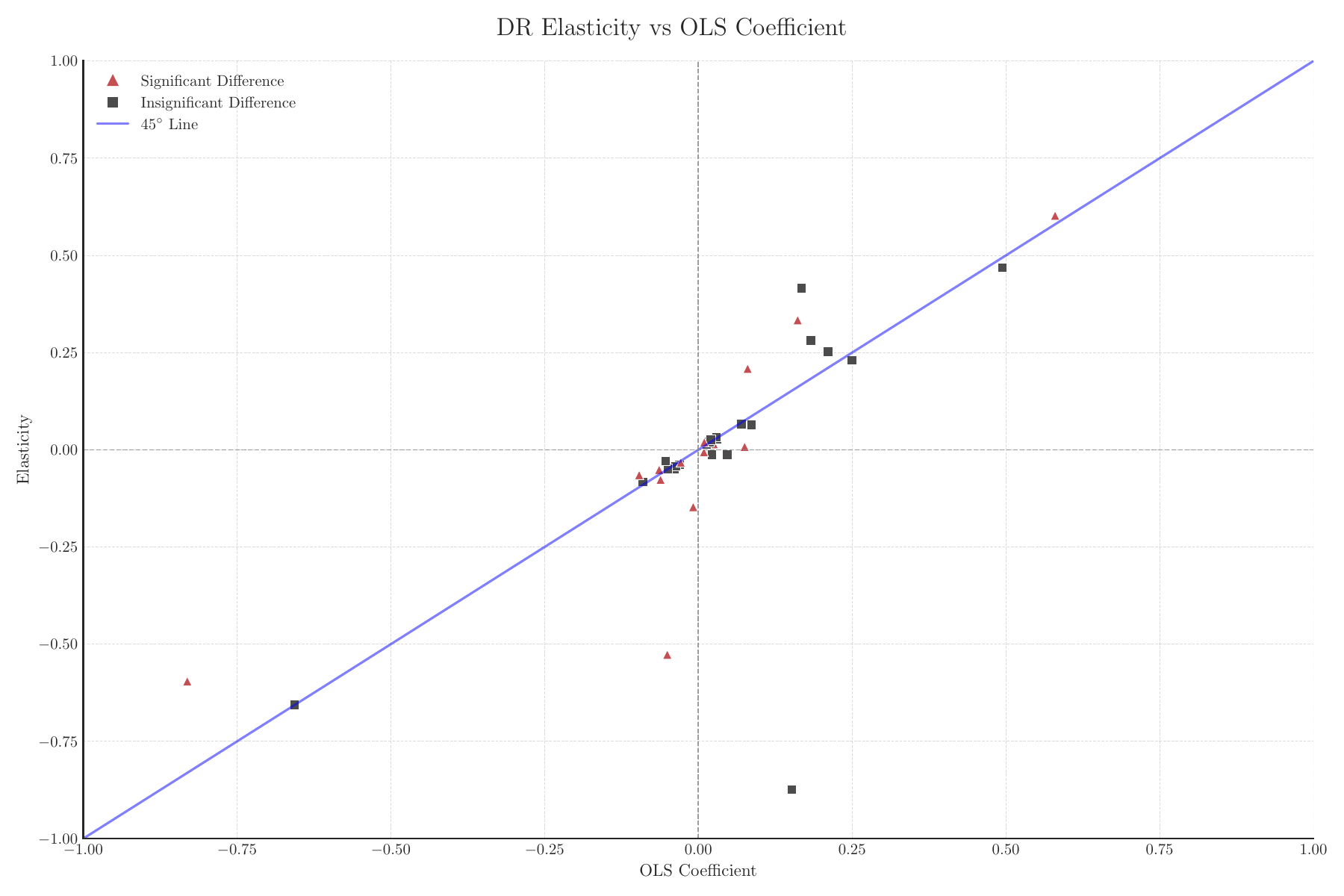}

}

\caption{\label{fig-elast-vs-ols-2d}}

\end{figure}%

For reference, we also compare elasticity estimates from PPML with those
from DREAM. While no papers in our sample used PPML.
Figure~\ref{fig-elast-vs-ppml-2d} presents these results. Here the
x-axis is the PPML estimate and the y-axis is the DREAM estimate. PPML
and DREAM report statistically significantly different estimates 38\% of
the time. There are two statistically significant sign reversals. As in
Figure~\ref{fig-elast-vs-ols-2d}, we can see that standard errors for
this test matter as well. Some estimates sit close to the parity line
but are statistically significantly different from each other, while
others are relatively distant from the parity line but are not
statistically significantly different.

\begin{figure}

\centering{

\includegraphics[width=1\linewidth,height=\textheight,keepaspectratio]{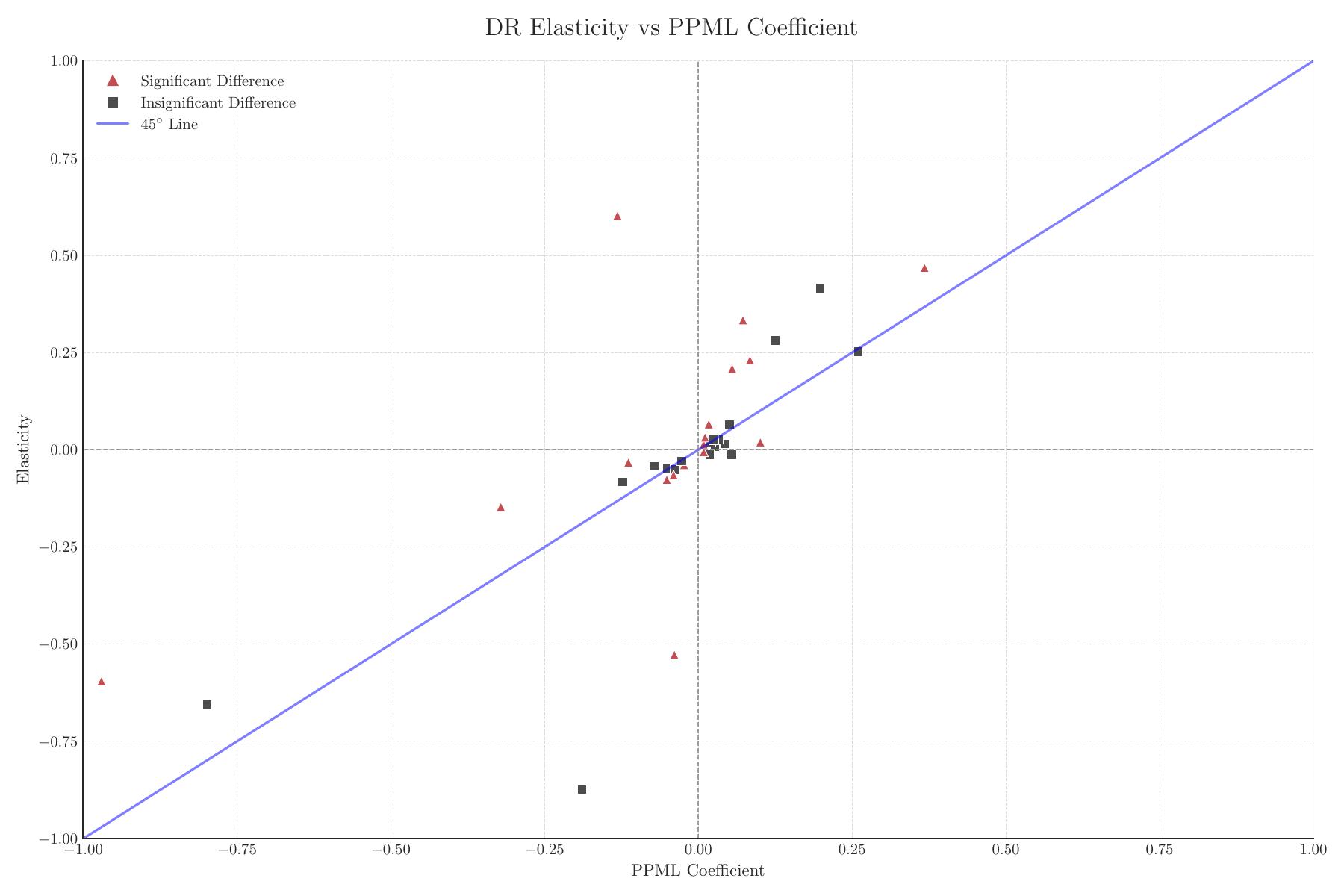}

}

\caption{\label{fig-elast-vs-ppml-2d}}

\end{figure}%

Table~\ref{tbl-differences} provides an aggregate summary of the
differences between the OLS estimator and DREAM, and the PPML estimator
and DREAM. Of the 50 results we test, 32 have no statistically
significant difference with OLS. Of those that are statistically
significantly different, 8 have an increase in estimate with DREAM (in
absolute terms), and 6 have a decrease (in absolute terms). Four
estimates reverse sign between OLS and DREAM. The comparison between
DREAM and PPML is similar. There are 31 estimates that are not
statistically significantly different, and 19 that are. Of the
statistically significantly different, 13 increase in absolute value
with DREAM, 4 decrease in absolute value, and 2 reverse signs.

\begin{longtable}[]{@{}
  >{\raggedright\arraybackslash}p{(\linewidth - 10\tabcolsep) * \real{0.0588}}
  >{\raggedright\arraybackslash}p{(\linewidth - 10\tabcolsep) * \real{0.1324}}
  >{\raggedright\arraybackslash}p{(\linewidth - 10\tabcolsep) * \real{0.2059}}
  >{\raggedright\arraybackslash}p{(\linewidth - 10\tabcolsep) * \real{0.2206}}
  >{\raggedright\arraybackslash}p{(\linewidth - 10\tabcolsep) * \real{0.2206}}
  >{\raggedright\arraybackslash}p{(\linewidth - 10\tabcolsep) * \real{0.1618}}@{}}
\caption{}\label{tbl-differences}\tabularnewline
\toprule\noalign{}
\begin{minipage}[b]{\linewidth}\raggedright
\end{minipage} & \begin{minipage}[b]{\linewidth}\raggedright
No Change
\end{minipage} & \begin{minipage}[b]{\linewidth}\raggedright
Sig. Different
\end{minipage} & \begin{minipage}[b]{\linewidth}\raggedright
Effect Increase
\end{minipage} & \begin{minipage}[b]{\linewidth}\raggedright
Effect Decrease
\end{minipage} & \begin{minipage}[b]{\linewidth}\raggedright
Sign Change
\end{minipage} \\
\midrule\noalign{}
\endfirsthead
\toprule\noalign{}
\begin{minipage}[b]{\linewidth}\raggedright
\end{minipage} & \begin{minipage}[b]{\linewidth}\raggedright
No Change
\end{minipage} & \begin{minipage}[b]{\linewidth}\raggedright
Sig. Different
\end{minipage} & \begin{minipage}[b]{\linewidth}\raggedright
Effect Increase
\end{minipage} & \begin{minipage}[b]{\linewidth}\raggedright
Effect Decrease
\end{minipage} & \begin{minipage}[b]{\linewidth}\raggedright
Sign Change
\end{minipage} \\
\midrule\noalign{}
\endhead
\bottomrule\noalign{}
\endlastfoot
OLS & 32 & 18 & 8 & 6 & 4 \\
PPML & 31 & 19 & 13 & 4 & 2 \\
\end{longtable}

Figure~\ref{fig-elast-vs-both-ci} present a direct comparison between
the DREAM, PPML, and OLS for the same result. The graph on the left
provides the point estimate and standard errors for estimates of the
difference in elasticities estimated from DREAM and OLS. These are
ordered in descending order, from most positive to most negative
difference. Note, you can't observe sign flips on this graph because the
points are just the differences in estimates. The results on the
right-hand panel correspond to those on the left. For example, the top
result shown for the Elasticity - OLS panel is the same result estimated
in the top of the Elasticity - PPML panel. Here we can see that
sometimes DREAM is statistically significantly different from both
estimates, while other times it is only statistically significantly
different from one or no estimates. DREAM disagrees with both PPML and
OLS in 11 out of the 50 regressions we re-estimated (9 out of 42 if we
exclude absolute DREAM estimates greater than 5). DREAM is not
statistically significantly different from OLS and PPML in 24 out of the
50 results (or 18 out of 42 in the restricted sample). DREAM is
significantly different from OLS and not from PPML in 7 results (no
difference with restricted sample). DREAM is significantly different
from PPML and not from OLS in 8 results (again, no difference with
restricted sample).

\begin{figure}

\centering{

\includegraphics[width=1\linewidth,height=\textheight,keepaspectratio]{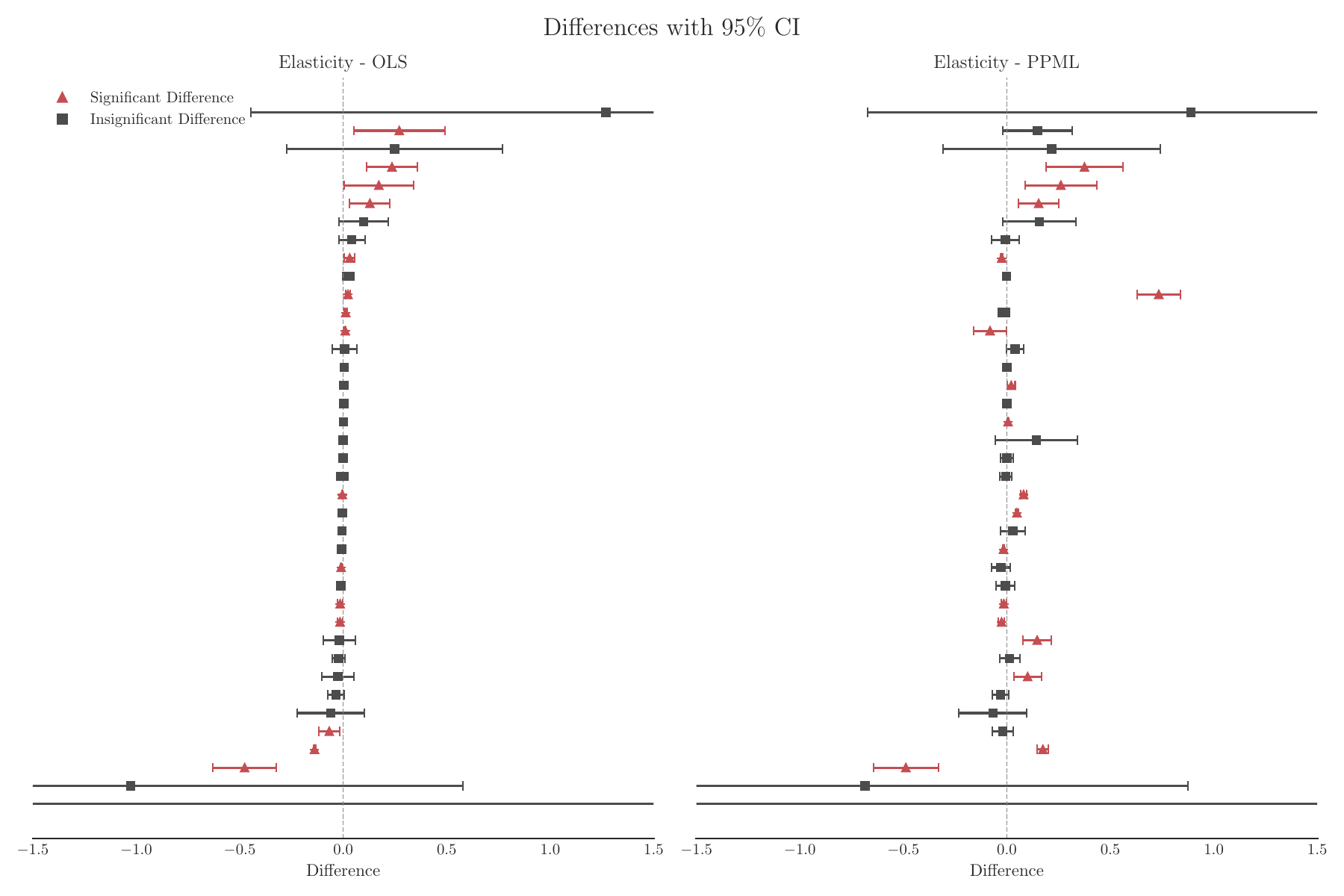}

}

\caption{\label{fig-elast-vs-both-ci}}

\end{figure}%

In total, this re-estimation of results from papers published in Top 5
journals in the year 2020 provide several takeaways. First, we show that
PPML and OLS estimates are frequently statistically significantly
different from each other, providing some suggestive evidence that
dependence may exist between the covariates and the error terms. Second,
we show that our estimator, which is consistent for the true
semi-elasticity, disagrees with OLS 38\% of the time, and disagrees with
PPML about 36\%. We see sign flips in 8\% of re-estimated results when
comparing DREAM to OLS, and 4\% when comparing DREAM to PPML.
Table~\ref{tbl-differences} shows that the differences between OLS and
DREAM, and PPML and DREAM estimates are remarkably consistent. When
directly comparing results from papers, we see that there are a mix of
outcomes, sometimes DREAM disagrees with both PPML and OLS in 22\% of
estimations, it disagrees with neither in 48\% of estimations, and
disagrees with one but not the other in 30\% of estimations.

\subsection{IV re-estimations}\label{iv-re-estimations}

Similarly to the OLS and DREAM case we re-estimate 19 results from 12
papers using the DREAM with control functions approach. Of these, 13
were significantly different at the 5\% level from the naive 2SLS
estimate. Six papers showed no difference when estimated using 2SLS or
IV-DREAM. Eight showed reduced effect size, of which five exhibited sign
reversal. Five showed increased effect size.

Among the statistically significant differences at the 5\% level, the
median absolute change was about 71\% of the original estimate.

\begin{figure}

\centering{

\includegraphics[width=1\linewidth,height=\textheight,keepaspectratio]{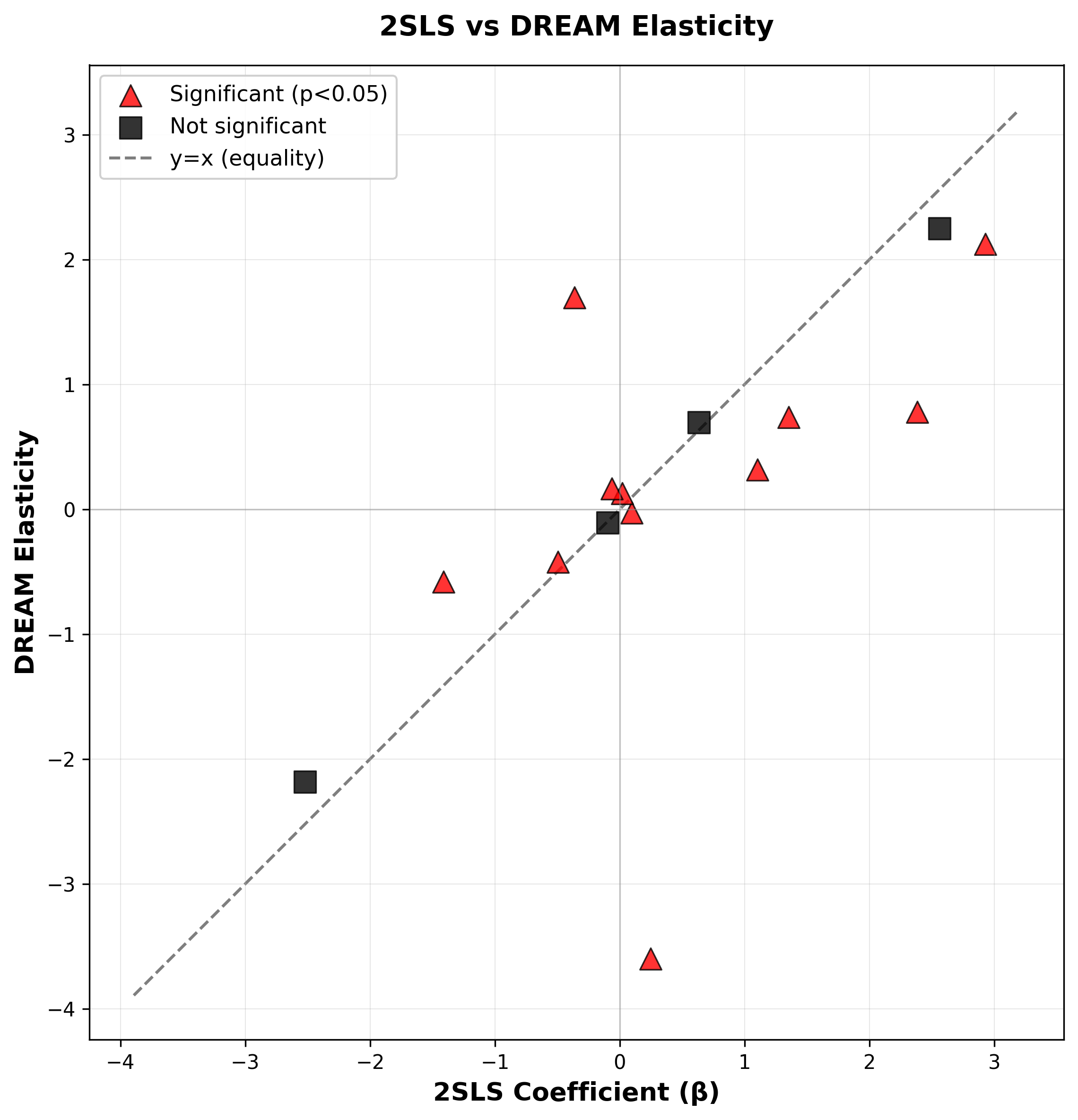}

}

\caption{\label{fig-2SLS-IV-DRNO}}

\end{figure}%

\section{Conclusion}\label{sec-conclusion}

Economists often use log dependent variable regressions to estimate
percentage changes. The coefficient estimates are consistent for the
geometric mean (semi-) elasticities. We give several examples where
economists specify estimands of arithmetic mean elasticities. We show
that arithmetic and geometric mean elasticities can be very different
when individual elasticities are heterogenous. We present them as a
sub-class of power-mean elasticities, and show that their difference can
be characterized using a weighted aggregation of the heterogenous
responses. We justify the estimation of power-mean elasticities by
appeal to a decision problem axiomatisation. We provide an estimator for
the average wedge between arithmetic and geometric mean elasticities.
Using this estimator, we replicate 50 papers, and find a median
difference of 67\%.

Future work could focus on eliciting the whole elasticity distribution,
clarifying the relationship between aggregation bias and
retransformation bias and further specifying which economic situations
demand which estimands.

Future drafts of this paper will include examples with randomised
controlled trials and discrete variables, for which we have the maths
but not yet the words.

\newpage
\bibliography{log-linear}

\newpage
\appendix

\section{Proofs}\label{proofs}

\subsection{\texorpdfstring{Proof of
\hyperref[myprp-noneqhet]{Proposition~\ref*{myprp-noneqhet}}}{Proof of }}\label{sec-noneqhetproof}

From \hyperref[mylem-represent]{Lemma~\ref*{mylem-represent}},

\[
\varepsilon_\phi(x)
=
\frac{\mathbb{E}\!\left[ Y(x)^\phi \, \varepsilon(\omega) \right]}
     {\mathbb{E}\!\left[ Y(x)^\phi \right]}.
\]

Using the outcome specification
\(\log Y(x;\omega)=a(\omega)+\varepsilon(\omega)\log x\), we can write
\[
Y(x;\omega)^\phi
=
\exp\!\left(\phi a(\omega)\right)\,x^{\phi \varepsilon(\omega)}.
\]

Substituting, \[
\varepsilon_\phi(x)
=
\frac{\mathbb{E}\!\left[ \varepsilon(\omega)\exp(\phi a(\omega))x^{\phi \varepsilon(\omega)} \right]}
     {\mathbb{E}\!\left[ \exp(\phi a(\omega))x^{\phi \varepsilon(\omega)} \right]}.
\]

For \(\phi = 0\), the weights collapse and \[
\varepsilon_0 = \mathbb{E}[\varepsilon(\omega)],
\] which is constant in \(x\).

For \(\phi \neq 0\), the weights depend on \(x\) through the term
\(x^{\phi \varepsilon(\omega)}\). If \(\varepsilon(\omega)\) is
non-degenerate, then different values of \(\varepsilon(\omega)\) are
differentially weighted as \(x\) changes. As a result,
\(\varepsilon_\phi(x)\) varies with \(x\) and cannot coincide with the
constant value \(\varepsilon_0\) except on a set of \(x\) of measure
zero.

Therefore, under heterogeneous responsiveness, aggregate elasticities
differ across aggregation rules and are generically non-constant
functions of \(x\) for all \(\phi \neq 0\).

\subsection{\texorpdfstring{Proof of
\hyperref[myprp-wedgedecomp]{Proposition~\ref*{myprp-wedgedecomp}}}{Proof of }}\label{sec-wedgedecompproof}

By definition of the power mean, for \(\phi \neq 0\), \[
\log M_\phi(x) = \frac{1}{\phi}\log \mathbb{E}[Y(x)^\phi] = \frac{1}{\phi}\log G_\phi(x).
\] Differentiating with respect to \(\log x\) yields \[
\varepsilon_\phi(x) \equiv \frac{d\log M_\phi(x)}{d\log x}
= \frac{1}{\phi}\frac{d \log G_\phi(x)}{d \log x}.
\]

Under \hyperref[asm-randomcoef]{Assumption~\ref*{asm-randomcoef}}, \[
Y(x;\omega)^\phi = \exp(\phi a(\omega))\,x^{\phi \varepsilon(\omega)}
= x^{\phi \varepsilon_0}\,\exp(\phi a(\omega))\,x^{\phi \tilde\varepsilon(\omega)}.
\] Therefore, \[
G_\phi(x)=\mathbb{E}[Y(x)^\phi]
= x^{\phi \varepsilon_0}\,\mathbb{E}\!\left[\exp(\phi a(\omega))\,x^{\phi \tilde\varepsilon(\omega)}\right].
\] Taking logs and differentiating, \[
\frac{1}{\phi}\frac{d\log G_\phi(x)}{d\log x}
=
\varepsilon_0
+
\frac{1}{\phi}\frac{d}{d\log x}\log \mathbb{E}\!\left[\exp(\phi a(\omega))\,x^{\phi \tilde\varepsilon(\omega)}\right],
\] which is the stated identity. The \(\phi=1\) case follows
immediately.

\subsection{\texorpdfstring{Proof of
\hyperref[myprp-representation]{Proposition~\ref*{myprp-representation}}}{Proof of }}\label{sec-proof-rep}

\textbf{Step 1: Individual Model (Axiom 1).}

From \hyperref[axm-scale-invariance]{Axiom~\ref*{axm-scale-invariance}},
\(Y(\lambda x) = g(\lambda) Y(x)\). Differentiating w.r.t. \(\lambda\)
and setting \(\lambda=1\) gives the differential equation
\(\frac{\partial \log Y}{\partial \log x} = g'(1) \equiv \varepsilon(\omega)\).
The solution is
\(Y(x; \omega) = \alpha(\omega) x^{\varepsilon(\omega)}\).

\textbf{Step 2: Quasi-Arithmetic Mean (Axioms 2 \& 3).}

The Kolmogorov-Nagumo Theorem states that any continuous aggregator
satisfying strict monotonicity, symmetry (Axiom 2), and decomposability
(Axiom 3) is a Quasi-Arithmetic Mean:
\(M(Y) = \psi^{-1}\left( \mathbb{E}[\psi(Y)] \right)\) for some
continuous monotonic \(\psi\).

\textbf{Step 3: Power Mean (Axiom 4).}

Imposing Homotheticity,
\(\psi^{-1}( \mathbb{E}[\psi(kY)] ) = k \psi^{-1}( \mathbb{E}[\psi(Y)] )\).
The only continuous solutions (up to affine transformation) are
\(\psi(y) = y^\phi\) (\(\phi \neq 0\)) and \(\psi(y) = \log y\)
(\(\phi = 0\)). Combining Step 1 and 3 yields the result.

\subsection{\texorpdfstring{Proof of
\hyperref[myprp-sufficiency]{Proposition~\ref*{myprp-sufficiency}}}{Proof of }}\label{sec-proof-sufficiency}

The DM solves \(\max_x \mathcal{V}(x, M_\phi(x))\). The First Order
Condition is:

\[\frac{\partial \mathcal{V}}{\partial x} + \frac{\partial \mathcal{V}}{\partial M_\phi} \frac{d M_\phi}{d x} = 0.\]

By definition,
\(\frac{d M_\phi}{d x} = M_\phi(x) \frac{\varepsilon_\phi(x)}{x}\).
Substituting this into the FOC:

\[\frac{\partial \mathcal{V}}{\partial x} + \frac{\partial \mathcal{V}}{\partial M_\phi} \frac{M_\phi(x)}{x} \varepsilon_\phi(x) = 0.\]

The terms \(\frac{\partial \mathcal{V}}{\partial x}\) and
\(\frac{\partial \mathcal{V}}{\partial M_\phi}\) depend on the policy
design and the aggregate level \(M_\phi\). The only additional term
depending on the micro-structure of heterogeneity is
\(\varepsilon_\phi(x)\). Thus, knowledge of \(F(\omega)\) is not
required; the scalar \(\varepsilon_\phi(x)\) is sufficient.

\subsection{Neyman Orthogonality Proof}\label{sec-noproof}

This appendix establishes the integration-by-parts identity assumed in
the main text and proves Neyman orthogonality of the score function
Theorem~\ref{thm-neymanorthogonality}.

\subsubsection{Integration by Parts
Identity}\label{integration-by-parts-identity}

We assume a conditional integration-by-parts identity: for any
sufficiently regular function \(q(x,z)\),

\[ \mathbb{E}\left[\nabla_x q(X,Z)\right] = -\mathbb{E}\left[s_0(X,Z) , q(X,Z)\right], \tag{IBP} \]

where \(s_0(x,z) := \nabla_x f_0(x \mid z) / f_0(x \mid z)\) is the
density score.

Sufficient conditions include (i) conditional compact support of
\(X \mid Z = z\) with \(q(\cdot, z)\) vanishing on the boundary, or (ii)
tail conditions ensuring the boundary term from integration by parts
vanishes.

\subsubsection{Moment Condition}\label{moment-condition}

\textbf{Lemma A.1.}
\(\mathbb{E}[\phi(W; \theta_0, \beta_0, \gamma_0, m_0, f_0)] = 0.\)

\emph{Proof.} By definition of \(\theta_0\),

\[ \mathbb{E}\left[\beta_0 + \frac{m_0'(X,Z)}{m_0(X,Z)} - \theta_0\right] = 0. \]

For the correction term, since \(\mathbb{E}[\delta \mid X, Z] = 0\),

\[ \mathbb{E}\left[\alpha_0(X,Z)(e^u - m_0(X,Z)) \mid X, Z\right] = \alpha_0(X,Z) , \mathbb{E}[\delta \mid X, Z] = 0. \]

Taking unconditional expectations yields the result.

\subsubsection{Proof of Neyman
Orthogonality}\label{proof-of-neyman-orthogonality}

\begin{proof}
We prove orthogonality with respect to each nuisance component. By
linearity of Gateaux derivatives, orthogonality in each direction
separately implies orthogonality for joint perturbations.

\paragraph{\texorpdfstring{Orthogonality with Respect to
\(m\)}{Orthogonality with Respect to m}}\label{orthogonality-with-respect-to-m}

Let \(m_t(x,z) = m_0(x,z) + t , h(x,z)\) be an admissible path with
\(m_t > 0\) for small \(|t|\). Define

\[ \Psi(t) := \mathbb{E}\left[\phi(W; \theta_0, \beta_0, \gamma_0, m_t, f_0)\right]. \]

We show \(\Psi'(0) = 0\).

\textbf{Step 1: Differentiate the \(m'/m\) term.} Since
\(m_t'/m_t = \nabla_x \log m_t\), we have
\[ \frac{d}{dt} \mathbb{E}\left[\frac{m_t'(X,Z)}{m_t(X,Z)}\right]\bigg|_{t=0} = \frac{d}{dt} \mathbb{E}\left[\nabla_x \log m_t(X,Z)\right]\bigg|_{t=0}. \]
Exchanging differentiation with respect to \(t\) and \(\nabla_x\)
(justified under regularity conditions in A2):
\[ = \mathbb{E}\left[\nabla_x \left(\frac{d}{dt} \log m_t(X,Z)\bigg|_{t=0}\right)\right]. \]
Now, \(\frac{d}{dt} \log(m_0 + th)\big|_{t=0} = h/m_0\), so
\[ \frac{d}{dt} \mathbb{E}\left[\frac{m_t'(X,Z)}{m_t(X,Z)}\right]\bigg|_{t=0} = \mathbb{E}\left[\nabla_x \left(\frac{h(X,Z)}{m_0(X,Z)}\right)\right]. \tag{D1} \]

\textbf{Step 2: Differentiate the correction term.}

With \(f_0\) fixed, \(\alpha_t(x,z) = -s_0(x,z)/m_t(x,z)\). The
correction term is

\[ \alpha_t(X,Z)(e^u - m_t(X,Z)) = -s_0(X,Z)\left(\frac{e^u}{m_t(X,Z)} - 1\right). \]

Differentiating and using \(\frac{d}{dt}(1/m_t)|_{t=0} = -h/m_0^2\),

\[ \frac{d}{dt}\left[-s_0(X,Z)\left(\frac{e^u}{m_t(X,Z)} - 1\right)\right]\bigg|_{t=0} = s_0(X,Z)  e^u  \frac{h(X,Z)}{m_0(X,Z)^2}. \]

Taking expectations and using \(\mathbb{E}[e^u \mid X, Z] = m_0(X,Z)\),

\[ \frac{d}{dt} \mathbb{E}\left[\alpha_t(X,Z)(e^u - m_t(X,Z))\right]\bigg|_{t=0} = \mathbb{E}\left[s_0(X,Z)  \frac{h(X,Z)}{m_0(X,Z)}\right]. \tag{D2} \]

\textbf{Step 3: Cancellation via integration by parts.}

Applying (IBP) with \(q(X,Z) = h(X,Z)/m_0(X,Z)\),

\[ \mathbb{E}\left[\nabla_x\left(\frac{h(X,Z)}{m_0(X,Z)}\right)\right] = -\mathbb{E}\left[s_0(X,Z) , \frac{h(X,Z)}{m_0(X,Z)}\right]. \tag{D3} \]

Combining (D1), (D2), and (D3):

\[ \Psi'(0) = -\mathbb{E}\left[s_0 , \frac{h}{m_0}\right] + \mathbb{E}\left[s_0 , \frac{h}{m_0}\right] = 0. \]

\paragraph{\texorpdfstring{Orthogonality with Respect to
\(f\)}{Orthogonality with Respect to f}}\label{orthogonality-with-respect-to-f}

Holding \(m_0\) fixed, the only dependence on \(f\) is through
\(\alpha\). For any path \(f_t\) through \(f_0\) with corresponding
\(\alpha_t\),

\[ \frac{d}{dt} \mathbb{E}\left[\alpha_t(X,Z)(e^u - m_0(X,Z))\right]\bigg|_{t=0} = \mathbb{E}\left[\dot{\alpha}_0(X,Z) \cdot \mathbb{E}[\delta \mid X, Z]\right] = 0, \]

where \(\dot{\alpha}_0\) denotes the Gateaux derivative at \(t = 0\).
The final equality uses \(\mathbb{E}[\delta \mid X, Z] = 0\).

\paragraph{\texorpdfstring{Orthogonality with Respect to
\((\beta, \gamma)\)}{Orthogonality with Respect to (\textbackslash beta, \textbackslash gamma)}}\label{orthogonality-with-respect-to-beta-gamma}

The parameters \((\beta, \gamma)\) enter only through
\(u(\beta, \gamma) = \log Y - \beta^\top X - \gamma^\top Z\).
Differentiating the correction term with respect to \(\beta\) in
direction \(b\):

\[ \frac{d}{dt} \mathbb{E}\left[\alpha_0(X,Z)\left(e^{u(\beta_0 + tb, \gamma_0)} - m_0(X,Z)\right)\right]\bigg|_{t=0} = -\mathbb{E}\left[\alpha_0(X,Z) , e^u , b^\top X\right]. \]

By iterated expectations and \(\mathbb{E}[e^u \mid X, Z] = m_0(X,Z)\):

\[ \mathbb{E}\left[\alpha_0(X,Z) , e^u , b^\top X\right] = \mathbb{E}\left[\alpha_0(X,Z) , m_0(X,Z) , b^\top X\right] = -\mathbb{E}\left[s_0(X,Z) , b^\top X\right], \]

using the definition \(\alpha_0 = -s_0/m_0\).

Meanwhile, the \(\beta\) term in the score contributes \(+b\) to the
derivative. Applying (IBP) with \(q(X,Z) = b^\top X\):

\[ \mathbb{E}[b] = b = -\mathbb{E}[s_0(X,Z) , b^\top X], \]

which is the identity \(\mathbb{E}[\nabla_x(b^\top X)] = b\). Hence the
two contributions cancel.

An analogous argument applies to \(\gamma\), completing the proof.
\end{proof}

\subsection{\texorpdfstring{Proof of
\hyperref[myprp-iv-impossibility]{Proposition~\ref*{myprp-iv-impossibility}}}{Proof of }}\label{proof-of-myprp-iv-impossibility}

This proof uses \(x\) instead of \(\log x\) but is not affected by this
change.

\begin{proof}
IV identifies \(\beta_0\) from
\(\text{Cov}(\log Y, Z)/\text{Cov}(X, Z)\). The arithmetic mean
semi-elasticity is
\[ \theta(x) = \frac{\mathbb{E}[\varepsilon_i , e^{a_i + \varepsilon_i x}]}{\mathbb{E}[e^{a_i + \varepsilon_i x}]}, \]
which depends on the full joint distribution of \((a_i, \varepsilon_i)\)
beyond first moments. The IV assumptions constrain only
\(\mathbb{E}[\varepsilon_i \mid Z]\) and \(\mathbb{E}[a_i \mid Z]\),
leaving higher moments --- and the joint dependence between \(a_i\) and
\(\varepsilon_i\) --- unrestricted.

We construct two observationally equivalent models with different
arithmetic semi-elasticities. Let \(X \in {0,1}\) and set
\(\alpha_0 = 0\) without loss.

\textbf{Model A} (homogeneous semi-elasticity, selection on levels). All
individuals share \(\varepsilon_i = \beta_0\). The intercept satisfies
\(a_i \mid X = 0 \sim N(0, 1)\) and
\(a_i \mid X = 1 \sim N(0, 1 + \sigma^2)\), with
\(\mathbb{E}[a_i \mid Z] = 0\). Heteroskedasticity arises from
selection: high-variance individuals choose treatment.

\textbf{Model B} (heterogeneous semi-elasticity, no selection). The
intercept and semi-elasticity are \(a_i \sim N(0, 1)\) and
\(\varepsilon_i \sim N(\beta_0, \sigma^2)\), mutually independent and
independent of \((X, Z)\).

Both models produce the same conditional distributions of observables.
At \(X = 0\): Model A gives \(\log Y \mid X!=!0 \sim N(0, 1)\); Model B
gives \(\log Y \mid X!=!0 = a_i \sim N(0, 1)\). At \(X = 1\): Model A
gives \(\log Y \mid X!=!1 = a_i + \beta_0\) with
\(a_i \mid X!=!1 \sim N(0, 1 + \sigma^2)\), so
\(\log Y \mid X!=!1 \sim N(\beta_0, 1 + \sigma^2)\); Model B gives
\(\log Y \mid X!=!1 = a_i + \varepsilon_i\) with
\(a_i + \varepsilon_i \sim N(\beta_0, 1 + \sigma^2)\) by independence.
Both satisfy \(\mathbb{E}[\varepsilon_i \mid Z] = \beta_0\) and
\(\mathbb{E}[a_i \mid Z] = 0\).

Yet the models imply different arithmetic mean semi-elasticities. In
Model A, \(\varepsilon_i = \beta_0\) for all individuals, so
\(\theta_A(x) = \beta_0\) for all \(x\). In Model B, \(a_i\) and
\(\varepsilon_i\) are independent normals, so
\(a_i + \varepsilon_i x \sim N(\beta_0 x,, 1 + \sigma^2 x^2)\). The
moment generating function gives
\[ \log \mathbb{E}[Y_B(x)] = \beta_0 x + \tfrac{1}{2}(1 + \sigma^2 x^2), \]
so \(\theta_B(x) = \beta_0 + \sigma^2 x\). At any \(x \neq 0\),
\(\theta_B(x) \neq \theta_A(x)\).
\end{proof}

\subsection{\texorpdfstring{Proof of
\hyperref[myprp-iv-exact-nonid]{Proposition~\ref*{myprp-iv-exact-nonid}}}{Proof of }}\label{proof-of-myprp-iv-exact-nonid}

This proof uses \(x\) instead of \(\log x\) but is not affected by this
change.

\begin{proof}
We begin with a lemma that provides compactly supported elements of the
null space of the Radon transform.

\begin{lemma}[Lemma (Directional derivative null
space)]\protect\hypertarget{lem-directional-derivative}{}\label{lem-directional-derivative}

For fixed \(x\), let \(\mathcal{R}_x\) denote the Radon projection along
lines \({(a, \varepsilon) : a + \varepsilon x = \ell}\). For any
\(\Psi \in C_c^\infty(\mathbb{R}^2)\), the directional derivative
\(D_x[\Psi] := (-x, \partial_a + \partial_\varepsilon)\Psi\) satisfies
\(\mathcal{R}_x[D_x\Psi] = 0\). Moreover,
\(D_x[\Psi] \in C_c^\infty(\mathbb{R}^2)\) with
\(\text{supp}(D_x[\Psi]) \subset \text{supp}(\Psi)\).

\end{lemma}

\begin{proof}
Parametrise the line \({a + \varepsilon x = \ell}\) by
\(s \mapsto (\ell - sx,, s)\). The chain rule gives
\(\frac{d}{ds}\Psi(\ell - sx, s) = -x,\partial_a\Psi + \partial_\varepsilon\Psi = (D_x\Psi)(\ell - sx, s)\),
so
\[ \mathcal{R}_x[D_x\Psi] = \int_{-\infty}^{\infty} \frac{d}{ds}\Psi(\ell - sx,, s), ds = \big[\Psi(\ell - sx,, s)\big]_{s = -\infty}^{s = +\infty} = 0, \]
where the boundary terms vanish by the compact support of \(\Psi\). That
\(D_x[\Psi]\) is smooth and supported inside \(\text{supp}(\Psi)\) is
immediate from the fact that differentiation preserves smoothness and
cannot enlarge the support.
\end{proof}

\emph{Proof of the Proposition.}

Let
\(k(x \mid a, \varepsilon, z) := f(X = x \mid a_i = a, \varepsilon_i = \varepsilon, Z = z)\)
denote the selection mechanism. The structural equation
\(\log Y = a + \varepsilon X\) and full independence imply that the
observable joint density satisfies
\[ f(\ell, x \mid z) = \int k(x \mid \ell - sx, s, z), f(\ell - sx, s), ds \tag{$\star$} \]
for all \((\ell, x, z)\), where \(\ell = \log y\). Define
\(h(a, \varepsilon, x, z) := k(x \mid a, \varepsilon, z), f(a, \varepsilon)\),
so (\(\star\)) states that for each fixed \((x,z)\), the observable
\(f(\ell, x \mid z)\) is the Radon projection of
\(h(\cdot, \cdot, x, z)\) along lines
\({(a, \varepsilon) : a + \varepsilon x = \ell}\). Normalisation of
\(k\) gives \(\int h(a, \varepsilon, x, z), dx = f(a, \varepsilon)\) for
all \((a, \varepsilon, z)\). The target
\(\theta(x) = \mathbb{E}[\varepsilon_i e^{a_i + \varepsilon_i x}] / \mathbb{E}[e^{a_i + \varepsilon_i x}]\)
depends on \(f\) alone: writing
\(M(x) := \mathbb{E}[e^{a + \varepsilon x}] = \int e^{a + \varepsilon x} f(a, \varepsilon), da, d\varepsilon\),
we have \(\theta(x) = M'(x)/M(x)\).

We construct two models \((f_0, k_0)\) and \((f_1, k_1)\) that produce
identical observables but yield \(\theta_1(x_0) \neq \theta_0(x_0)\).

\emph{Construction of the perturbation \(\delta f\).} We work with a
product form
\(\delta f(a, \varepsilon) = \rho(a), \varphi(\varepsilon)\). Let
\(\rho \in C_c^\infty(\mathbb{R})\) with \(\rho > 0\) on some interval
and \(\int \rho, da > 0\). We seek
\(\varphi \in C_c^\infty(\mathbb{R})\) satisfying four conditions:
\[ \int \varphi = 0, \quad \int \varepsilon, \varphi = 0, \quad \int e^{\varepsilon x_0}, \varphi \neq 0, \quad \frac{\int \varepsilon, e^{\varepsilon x_0}, \varphi}{\int e^{\varepsilon x_0}, \varphi} \neq \theta_0(x_0). \tag{1} \]
Such \(\varphi\) exists. Choose four smooth nonnegative bumps
\(\eta_1, \ldots, \eta_4 \in C_c^\infty(\mathbb{R})\) with disjoint
supports centred at distinct points
\(\varepsilon_1 < \varepsilon_2 < \varepsilon_3 < \varepsilon_4\) in a
compact interval \(I\), and set \(\varphi = \sum_{j=1}^4 c_j \eta_j\).
The first two conditions form a homogeneous \(2 \times 4\) linear system
in \((c_1, \ldots, c_4)\) whose coefficient matrix
\([\int \eta_j;\ \int \varepsilon,\eta_j]\) has rank 2 (since the
\(\varepsilon_j\) are distinct), so the solution space is 2-dimensional.
Define \(D_j := \int e^{\varepsilon x_0} \eta_j > 0\) and
\(N_j := \int \varepsilon, e^{\varepsilon x_0} \eta_j\). On the 2D
solution space, the third condition excludes the hyperplane
\({\sum c_j D_j = 0}\) (a line in 2D), and the fourth excludes
\({\sum c_j N_j = \theta_0(x_0) \sum c_j D_j}\) (another line). Since
\(N_j/D_j\) is approximately \(\varepsilon_j\) (the
\(e^{\varepsilon x_0}\)-weighted mean of \(\varepsilon\) under
\(\eta_j\)) and these are distinct, not all \(N_j/D_j\) equal
\(\theta_0(x_0)\), so the two excluded lines are distinct and do not
cover the 2D space. A point avoiding both gives the required \((c_j)\),
hence \(\varphi\).

Set
\(\delta f(a, \varepsilon) := \rho(a), \varphi(\varepsilon) \in C_c^\infty(\mathbb{R}^2)\).
By the product structure,
\(\int \delta f = (\int \rho)(\int \varphi) = 0\),
\(\int \varepsilon, \delta f = (\int \rho)(\int \varepsilon, \varphi) = 0\),
and
\(\int e^{a + \varepsilon x_0}, \delta f = (\int e^a \rho)(\int e^{\varepsilon x_0} \varphi) \neq 0\)
since both factors are nonzero. Define \(f_1 := f_0 + t, \delta f\).
Since \(f_0 > 0\) is continuous on \(\mathbb{R}^2\) and
\(\text{supp}(\delta f)\) is compact,
\(\inf_{\text{supp}(\delta f)} f_0 > 0\) by the extreme value theorem,
so \(f_1 > 0\) for
\(|t| < \inf_{\text{supp}(\delta f)} f_0 / |\delta f|_\infty\).

\emph{Compactly supported null-space construction.} We construct a
family \({n_x}\) satisfying three properties:
\(n_x \in \ker \mathcal{R}_x\) for each \(x\),
\(\int n_x, dx = \delta f\), and each \(n_x\) is compactly supported in
\((a, \varepsilon)\).

Define the antiderivative
\(\Phi(\varepsilon) := \int_{-\infty}^{\varepsilon} \varphi(s), ds\).
Since \(\varphi \in C_c^\infty\) and \(\int \varphi = 0\), the function
\(\Phi\) is compactly supported: below \(\text{supp}(\varphi)\) the
integrand vanishes, and above \(\text{supp}(\varphi)\) the total
integral is zero. Moreover \(\Phi \in C_c^\infty(\mathbb{R})\) with
\(\Phi' = \varphi\) and
\(\text{supp}(\varphi) \subset \text{supp}(\Phi)\).

Set
\(\Psi_0(a, \varepsilon) := \rho(a), \Phi(\varepsilon) \in C_c^\infty(\mathbb{R}^2)\),
with
\(\text{supp}(\Psi_0) = \text{supp}(\rho) \times \text{supp}(\Phi) =: S\).
Let \(p \in C_c^\infty(\mathbb{R})\) be a (possibly signed) function
with \(\int p = 1\), \(\int x, p(x), dx = 0\), and compact support
\(I_p \subset \text{supp}(X)\). Such \(p\) exists on any interval of
positive length: take two smooth bumps at distinct points in \(I_p\) and
solve the \(2 \times 2\) system for the coefficients. Define
\[ n_x(a, \varepsilon) := p(x), D_x[\Psi_0](a, \varepsilon) = p(x)\big[-x, \rho'(a), \Phi(\varepsilon) + \rho(a), \varphi(\varepsilon)\big]. \tag{2} \]

The three required properties now follow. For the null-space property:
for each fixed \(x\), \(n_x = p(x) \cdot D_x[\Psi_0]\), and the Lemma
gives \(\mathcal{R}_x[D_x\Psi_0] = 0\); by linearity,
\(\mathcal{R}_x[n_x] = p(x) \cdot 0 = 0\). For the integral condition:
integrating (2) over \(x\),
\[ \int n_x, dx = -\rho'(a),\Phi(\varepsilon) \underbrace{\int x, p(x), dx}_{= 0} + \rho(a),\varphi(\varepsilon) \underbrace{\int p(x), dx}_{= 1} = \rho(a), \varphi(\varepsilon) = \delta f(a, \varepsilon). \]
Fubini is justified since \(p\) has compact support and the integrand is
smooth with support in the fixed compact set \(S\) for each \(x\). For
the compact support property: since
\(\text{supp}(\rho') \subset \text{supp}(\rho)\) and
\(\text{supp}(\varphi) \subset \text{supp}(\Phi)\), both terms in (2)
are supported in \(S\) for each \(x\), and \(n_x = 0\) for
\(x \notin I_p\).

\emph{Observational equivalence.} Define
\(h_1(a, \varepsilon, x, z) := h_0(a, \varepsilon, x, z) + t, n_x(a, \varepsilon)\).
The perturbation \(n_x\) does not depend on \(z\). The Radon transform
is linear and \(n_x \in \ker \mathcal{R}_x\), so
\[ \mathcal{R}_x[h_1] = \mathcal{R}_x[h_0] + t \cdot \underbrace{\mathcal{R}_x[n_x]}_{= 0} = f(\ell, x \mid z) \]
for every \((\ell, x, z)\); this is exact, not a first-order
approximation. The marginal condition holds since
\(\int h_1, dx = \int h_0, dx + t \int n_x, dx = f_0 + t, \delta f = f_1\).

\emph{Positivity.} The function \(n_x(a, \varepsilon)\) is supported in
the compact set \(S \times I_p\) (in \((a, \varepsilon, x)\)-space) and
does not depend on \(z\). On \(S \times I_p\): for each
\((a, \varepsilon, x) \in S \times I_p\), the hypothesis
\(\inf_{z} k_0(x \mid a, \varepsilon, z) > 0\) together with
\(f_0(a, \varepsilon) > 0\) gives
\(\inf_{z} h_0(a, \varepsilon, x, z) > 0\). Since this infimum is a
positive lower-semicontinuous function on the compact set
\(S \times I_p\), it attains a positive minimum:
\[ c := \inf_{(a,\varepsilon) \in S,; x \in I_p,; z} h_0(a, \varepsilon, x, z) > 0. \]
Let
\(\overline{M} := \sup_{(a,\varepsilon,x)} |n_x(a, \varepsilon)| < \infty\)
(continuous on a compact set). Then for all \(z\),
\[ h_1(a, \varepsilon, x, z) = h_0(a, \varepsilon, x, z) + t, n_x(a, \varepsilon) \geq c - |t| \overline{M} > 0 \quad \text{for } |t| < c/\overline{M}. \]
Outside \(S \times I_p\): either \((a, \varepsilon) \notin S\) or
\(x \notin I_p\), so \(n_x(a, \varepsilon) = 0\) and \(h_1 = h_0 > 0\)
(since \(f_0 > 0\) and \(k_0 > 0\)). Setting \(k_1 := h_1 / f_1\): since
\(h_1 > 0\) and \(f_1 > 0\), \(k_1\) is well-defined and positive, and
\(\int k_1, dx = (1/f_1)\int h_1, dx = f_1/f_1 = 1\) (using that \(f_1\)
does not depend on \(x\)), so \(k_1\) is a valid conditional density.

\emph{Change in \(\theta\).} The two models \((f_0, k_0)\) and
\((f_1, k_1)\) produce identical observables. It remains to show
\(\theta_1(x_0) \neq \theta_0(x_0)\). Since \(\theta(x) = M'(x)/M(x)\)
where \(M(x) = \int e^{a + \varepsilon x} f, da, d\varepsilon\), write
\(M_j = \int e^{a+\varepsilon x_0} f_j\) and
\(N_j = \int \varepsilon, e^{a+\varepsilon x_0} f_j\) for \(j = 0, 1\),
so \(\theta_j(x_0) = N_j / M_j\). Then
\[ \theta_1(x_0) = \theta_0(x_0) \iff \frac{N_0 + t,\delta N}{M_0 + t,\delta M} = \frac{N_0}{M_0} \iff \frac{\delta N}{\delta M} = \frac{N_0}{M_0} = \theta_0(x_0), \]
where
\(\delta M = \int e^{a+\varepsilon x_0}\delta f = (\int e^a \rho)(\int e^{\varepsilon x_0}\varphi)\)
and
\(\delta N = \int \varepsilon, e^{a + \varepsilon x_0}\delta f = (\int e^a \rho)(\int \varepsilon, e^{\varepsilon x_0}\varphi)\),
so
\(\delta N / \delta M = \int \varepsilon, e^{\varepsilon x_0}\varphi / \int e^{\varepsilon x_0}\varphi\).
The fourth condition in (1) was precisely
\(\delta N/\delta M \neq \theta_0(x_0)\), so
\(\theta_1(x_0) \neq \theta_0(x_0)\).
\end{proof}

\begin{refremark}[Contrast with triangularity]
Under a triangular first stage \(X = g(Z, V)\) with
\(V \perp\!\!\perp Z\) and \(g\) monotone in \(V\), observing
\((X = x, Z = z)\) pins down \(V = g^{-1}(z, x)\). Conditional on
\(V = v\), the selection mechanism is degenerate (\(k\) places all mass
at \(x = g(z, v)\)), and the observable at each \((x, z)\) is a Radon
projection of the single conditional density
\(f(a, \varepsilon \mid V!=!v)\). For each fixed \(v\), varying \(z\)
produces projections at multiple angles (since \(x = g(z, v)\) varies
with \(z\)); provided \(g(\cdot, v)\) generates a sufficient range of
treatment values, the Radon transform is injective on
\(L^2(\mathbb{R}^2)\), so \(f(a, \varepsilon \mid V\!=\!v)\) is
identified for each \(v\), and integrating over \(f_V\) recovers
\(\mathbb{E}[e^{a + \varepsilon x}]\). Without triangularity,
\(h(\cdot, \cdot, x, z)\) is a different unknown function for each
\((x, z)\) --- we observe one projection of each, rather than many
projections of one --- and the construction above shows that
reconstruction fails.

\label{rem-triangularity-exact}

\end{refremark}

\subsection{\texorpdfstring{Proof of
\hyperref[myprp-control-function-identification]{Proposition~\ref*{myprp-control-function-identification}}}{Proof of }}\label{proof-of-myprp-control-function-identification}

This proof uses \(x\) instead of \(\log x\) but is not affected by this
change.

\begin{proof}
Under the triangular first stage \(X = g(Z, V)\) with \(g\) strictly
monotone in \(V\) and \(V \perp\!\!\perp Z\), the control variable
\(V = g^{-1}(Z, X)\) is identified from observed data. The exogeneity
condition \((a_i, \varepsilon_i) \perp\!\!\perp Z \mid V\) ensures that,
conditional on \(V = v\), the structural equation
\(\log Y = a + \varepsilon X\) generates observables whose distribution
depends only on \(f(a, \varepsilon \mid V = v)\) and the deterministic
assignment \(X = g(Z, v)\).

For each fixed \(v\), varying \(Z = z\) moves \(X = g(z, v)\) across the
set \(\mathcal{X}(v) := {g(z, v) : z \in \text{supp}(Z)}\). At each such
\(x \in \mathcal{X}(v)\), the conditional density of \(\log Y\) given
\(X = x, V = v\) is the Radon projection of
\(f(a, \varepsilon \mid V = v)\) along lines
\({(a, \varepsilon) : a + \varepsilon x = \ell}\):
\[ f_{\log Y \mid X, V}(\ell \mid x, v) = \int f(\ell - sx, s \mid V = v), ds = \mathcal{R}_x[f(\cdot, \cdot \mid V = v)]). \]

By the sufficient support assumption, \(\mathcal{X}(v)\) has nonempty
interior for each \(v\). The Radon transform on \(\mathbb{R}^2\) is
injective on \(L^2(\mathbb{R}^2)\) when projections are available at a
set of angles (equivalently, slopes \(x\)) containing an interval. Since
\(f(\cdot, \cdot \mid V = v) \in L^2(\mathbb{R}^2)\) (implied by the
finiteness of \(\mathbb{E}[e^{a + \varepsilon x}]\)), the conditional
density \(f(a, \varepsilon \mid V = v)\) is identified for each \(v\).

With \(f(a, \varepsilon \mid V = v)\) identified for each \(v\) and
\(f_V\) identified from the marginal of \(V\),
\[ \mathbb{E}[e^{a + \varepsilon x}] = \int \left(\int e^{a + \varepsilon x}, f(a, \varepsilon \mid V = v), da, d\varepsilon\right) f_V(v), dv \]
is identified for all \(x\) in the support. Similarly,
\(\mathbb{E}[\varepsilon, e^{a + \varepsilon x}]\) is identified by the
same integration with the additional \(\varepsilon\) factor in the inner
integral. Hence
\(\theta(x) = \mathbb{E}[\varepsilon, e^{a + \varepsilon x}] / \mathbb{E}[e^{a + \varepsilon x}]\)
is identified.

Note that the argument places no restriction on how the joint of
\((a_i, \varepsilon_i)\) depends on \(V\). In particular,
\(\varepsilon_i\) may be correlated with \(V\) and identification is
unaffected because \(f(a, \varepsilon \mid V = v)\) is recovered in full
for each \(v\).
\end{proof}

\subsection{Binary Treatment PPML is equivalent to Manning 1998
correction estimator.}\label{sec-manningppml}

\begin{proof}
Take a finite sample

\[
(y_{i}, x_{i})
\]

With \(x_i\) binary.

Consider the PPML estimator of \(y_i\) on a constant and \(x_i\)

Model:

\[
\mathbb{E}[y_{i}|x_{i}] = e^{ \gamma_{0} + \gamma_{1} x_{i}}
\]

It stands to reason that

\[
e^{\hat{\gamma_{0}}} = \bar{y_{0}}
\]

and

\[
e^{ \hat{\gamma}_{1}} = \frac{\bar{y_{1}}}{\bar{y_{0}}}
\]

where \[
\bar{y_{i}} = \frac{1}{n}\sum_{i=1}^n y_{i} \mathcal{I}({x_{i} = i})
\]

Now, consider the OLS estimator of \(\log{y_{i}}\) on a constant and
\(x_i\).

\[
\mathbb{E}[\log{y_{i}}|x_{i}] = \beta_{0} + \beta_{1} x_{i}
\]

By method of moments:

\[
\hat{\beta_{0}} = \bar{\log{y_{0}}}
\] and \[
\hat{\beta_{1}} = \bar{\log{y_{1}}} - \bar{\log{y_{0}}}
\]

where \[
\bar{\log{y_{i}}} = \frac{1}{n} \sum_{i=1}^{n} \log{y_{i}} \mathcal{I}(x_{i} = i)
\]

Then exponentiated OLS residuals are given by

\[
e^{\hat{u_{i}}} = \begin{cases}
y_{i}e^{ -\bar{\log y_{0} }} \ \text{if } x_{i} = 0  \\ 
y_{i}e^{ -\bar{\log y_{1} }} \ \text{if } x_{i} = 1
\end{cases} 
\]

The corrected estimator then is

\[
e^{ \hat{\beta_{1}}} \frac{\bar{y_{1}} e^{ -\bar{\log{y_{1}}} }}{\bar{y_{0}}e^{\bar{-\log{y_{0}}}}} - 1 = e^{\bar{\log{y_{1}}} - \bar{\log{y_{0}}}} \frac{\bar{y_{1}} e^{ -\bar{\log{y_{1}}} }}{\bar{y_{0}}e^{\bar{-\log{y_{0}}}}} - 1 = \frac{\bar{y_{1}}}{\bar{y_{0}}} - 1 =  e^{\hat{\gamma_{1}}} - 1
\]

Thus, PPML in levels is equivalent to the
\citet{manningLoggedDependentVariable1998} correction estimator with a
binary variable.
\end{proof}

\subsection{Proofs for IV Debiased Estimators}\label{sec-iv-proof}

\subsubsection{Preliminaries}\label{preliminaries}

We establish Neyman orthogonality for the score
\[ \phi(W; \theta, \eta) = \beta + \frac{\mu'(X)}{\mu(X)} - \theta - \frac{\omega(X,V) S_X(X)}{\mu(X)}\left(e^u - m(X,V)\right) - \lambda(Z)^\top(X - g(Z)), \]
where \(\eta = (\beta, \rho, m, g, \omega, S_X, \lambda)\) collects the
nuisance functions. The residual \(R(X,V) := e^u - m_0(X,V)\) satisfies
\(\mathbb{E}[R \mid X, V] = 0\) by construction.

\subsubsection{Orthogonality with respect to the outcome
nuisance}\label{orthogonality-with-respect-to-the-outcome-nuisance}

\begin{lemma}[]\protect\hypertarget{lem-orthog-m}{}\label{lem-orthog-m}

The score \(\phi\) is Neyman-orthogonal with respect to the outcome
nuisance \(m\), that is,
\(\partial_m \mathbb{E}[\phi]\big|_{m_0}(\delta_m) = 0\) for all
admissible perturbations \(\delta_m\).

\end{lemma}

\begin{proof}
Define the functional
\[ J(m) := \mathbb{E}_X\left[\frac{\mu'(X)}{\mu(X)}\right], \qquad \mu(x) = \mathbb{E}_V[m(x, V)]. \]
For the perturbation \(m_t = m_0 + t\delta_m\), we have
\(\mu_t(x) = \mu_0(x) + t\Delta\mu(x)\) where
\(\Delta\mu(x) = \mathbb{E}_V[\delta_m(x, V)]\). The Gateaux derivative
is
\[ \partial_m J(\delta_m) = \mathbb{E}_X\left[\frac{\Delta\mu'(X)}{\mu_0(X)} - \frac{\mu_0'(X)\Delta\mu(X)}{\mu_0(X)^2}\right]. \]
Integrating by parts on the first term, with boundary terms vanishing by
regularity, yields
\[ \mathbb{E}_X\left[\frac{\Delta\mu'(X)}{\mu_0(X)}\right] = -\mathbb{E}_X\left[\frac{\Delta\mu(X)}{\mu_0(X)}\left(S_X(X) - \frac{\mu_0'(X)}{\mu_0(X)}\right)\right]. \]
Combining terms gives
\[ \partial_m J(\delta_m) = -\mathbb{E}_X\left[\frac{S_X(X)\Delta\mu(X)}{\mu_0(X)}\right] = -\mathbb{E}_{X,V}\left[\frac{\omega_0(X,V) S_X(X)}{\mu_0(X)}\delta_m(X,V)\right], \]
where the final equality uses the identity
\(\mathbb{E}_X[h(X)\mathbb{E}_V[k(X,V)]] = \mathbb{E}_{X,V}[\omega_0(X,V)h(X)k(X,V)]\).
The Riesz representer for \(m\) is therefore
\(\alpha_0(x,v) = -\omega_0(x,v) S_X(x)/\mu_0(x)\).

The derivative of the debiasing term
\(\mathbb{E}[\alpha_0(X,V)(e^u - m(X,V))]\) with respect to \(m\) equals
\(-\mathbb{E}_{X,V}[\alpha_0(X,V)\delta_m(X,V)]\). The total derivative
is then
\[ \partial_m \mathbb{E}[\phi] = \partial_m J(\delta_m) + \mathbb{E}[\alpha_0 \delta_m] = \mathbb{E}[\alpha_0 \delta_m] - \mathbb{E}[\alpha_0 \delta_m] = 0. \]
\end{proof}

\subsubsection{Orthogonality with respect to the first-stage
nuisance}\label{orthogonality-with-respect-to-the-first-stage-nuisance}

\begin{lemma}[]\protect\hypertarget{lem-orthog-g}{}\label{lem-orthog-g}

The score \(\phi\) is Neyman-orthogonal with respect to the first-stage
nuisance \(g\), that is,
\(\partial_g \mathbb{E}[\phi]\big|_{g_0}(\delta_g) = 0\) for all
admissible perturbations \(\delta_g\).

\end{lemma}

\begin{proof}
Consider the perturbation \(g_t = g_0 + t\delta_g\), which induces
\(V_t = X - g_t(Z) = V - t\delta_g(Z)\). This affects \(\phi\) through
evaluation points in \(m\), \(\omega\), and \(\mu\); through the density
ratio \(\omega_t\); and through the first-stage correction term.

The derivative of the main term \(J(m)\) with respect to \(g\) arises
from the shift in \(V\), giving
\[ \partial_g \mathbb{E}\left[\frac{\mu'(X)}{\mu(X)}\right] = \mathbb{E}\left[\nabla_v\left(\frac{\mu_0'(X)}{\mu_0(X)}\right) \cdot (-\delta_g(Z))\right]. \]
Since \(\mu_0(x) = \mathbb{E}_V[m_0(x,V)]\), this term involves
\(\bar{m}_v(x) := \mathbb{E}_V[\partial_v m_0(x,V)]\). The derivative of
the debiasing term \(-\mathbb{E}[\alpha_0(X,V)R(X,V)]\) includes
contributions from shifting the evaluation point in \(\alpha_0\) and
\(R\), and from the perturbation of the density ratio \(\omega_t\).

Collecting all terms proportional to \(\delta_g(Z)\), we obtain
\[ \partial_g \mathbb{E}[\phi] = \mathbb{E}[\Lambda(X,V,Z)^\top \delta_g(Z)] - \mathbb{E}[\lambda_0(Z)^\top \delta_g(Z)], \]
where \(\Lambda(X,V,Z)\) aggregates the pathwise derivatives. Setting
\(\lambda_0(z) := \mathbb{E}[\Lambda(X,V,Z) \mid Z = z]\) and applying
the tower property yields
\[ \partial_g \mathbb{E}[\phi] = \mathbb{E}[\Lambda^\top \delta_g(Z)] - \mathbb{E}[\mathbb{E}[\Lambda \mid Z]^\top \delta_g(Z)] = 0. \]
In practice, we estimate \(\lambda_0\) via automatic debiasing rather
than computing the explicit form of \(\Lambda\).
\end{proof}

\subsubsection{Orthogonality with respect to auxiliary
nuisances}\label{orthogonality-with-respect-to-auxiliary-nuisances}

\begin{lemma}[]\protect\hypertarget{lem-orthog-auxiliary}{}\label{lem-orthog-auxiliary}

The score \(\phi\) is Neyman-orthogonal with respect to the auxiliary
nuisances \((\omega, S_X, \lambda)\).

\end{lemma}

\begin{proof}
For the density ratio \(\omega\), the only relevant term is
\(-\mathbb{E}[\omega(X,V) S_X(X) R(X,V) / \mu_0(X)]\). The derivative
with respect to \(\omega\) is
\[ \partial_\omega \mathbb{E}[\phi] = -\mathbb{E}\left[\frac{\delta_\omega(X,V) S_X(X) R(X,V)}{\mu_0(X)}\right]. \]
Since \(\mathbb{E}[R(X,V) \mid X, V] = 0\) and all other terms are
\((X,V)\)-measurable, this equals zero. The same argument applies to
perturbations in \(S_X\).

For the Riesz representer \(\lambda\), the derivative is
\[ \partial_\lambda \mathbb{E}[\phi] = -\mathbb{E}[\delta_\lambda(Z)^\top (X - g_0(Z))] = -\mathbb{E}[\delta_\lambda(Z)^\top V]. \]
By iterated expectations and \(Z \perp V\), we have
\(\mathbb{E}[\delta_\lambda(Z)^\top V] = \mathbb{E}[\delta_\lambda(Z)^\top \mathbb{E}[V \mid Z]] = 0\).
\end{proof}

\subsubsection{Orthogonality with respect to finite-dimensional
parameters}\label{orthogonality-with-respect-to-finite-dimensional-parameters}

\begin{lemma}[]\protect\hypertarget{lem-orthog-beta-rho}{}\label{lem-orthog-beta-rho}

The score \(\phi\) is Neyman-orthogonal with respect to the
finite-dimensional parameters \((\beta, \rho)\).

\end{lemma}

\begin{proof}
The parameters enter through
\(u(\beta, \rho) = \log Y - \beta^\top X - \rho^\top V\) in the residual
\(e^u - m_0(X,V)\). The derivative with respect to \(\beta\) is
\[ \partial_\beta \mathbb{E}[\phi] = I + \mathbb{E}\left[\alpha_0(X,V) \cdot e^{u_0} \cdot (-X)\right], \]
where \(I\) is the identity from \(\partial_\beta \beta = I\). At truth,
\(e^{u_0} = Y e^{-\beta_0^\top X - \rho_0^\top V}\), and
\[ \mathbb{E}[\alpha_0(X,V) e^{u_0} X] = \mathbb{E}[\alpha_0 \cdot \mathbb{E}[e^{u_0} X \mid X,V]] = \mathbb{E}[\alpha_0 \cdot X \cdot m_0(X,V)]. \]
The first-order condition for \(\beta_0\) in the control function
regression implies this equals \(I\), yielding
\(\partial_\beta \mathbb{E}[\phi] = 0\). An analogous argument applies
for \(\rho\).
\end{proof}

\subsubsection{\texorpdfstring{Proof of
Theorem~\ref{thm-iv-neymanorthogonality}}{Proof of Theorem~}}\label{proof-of-thm-iv-neymanorthogonality}

\begin{proof}
Combine Lemma~\ref{lem-orthog-m}, Lemma~\ref{lem-orthog-g},
Lemma~\ref{lem-orthog-auxiliary}, and Lemma~\ref{lem-orthog-beta-rho}.
\end{proof}

\subsubsection{\texorpdfstring{Proof of
Theorem~\ref{thm-iv-asympnorm}}{Proof of Theorem~}}\label{proof-of-thm-iv-asympnorm}

\begin{proof}
By Neyman orthogonality (Theorem~\ref{thm-iv-neymanorthogonality}), the
score satisfies
\(\partial_\eta \mathbb{E}[\phi(W; \theta_0, \eta)]\big|_{\eta_0} = 0\).
Under assumption B5, the product-rate conditions ensure
\[ \left|\mathbb{E}[\phi(W; \theta_0, \hat{\eta})] - \mathbb{E}[\phi(W; \theta_0, \eta_0)]\right| = O_p(|\hat{\eta} - \eta_0|^2) = o_p(n^{-1/2}). \]
The moment condition
\(\partial_\theta \mathbb{E}[\phi(W; \theta)]\big|_{\theta_0} = -1\)
implies the influence function equals \(\phi_0(W)\). By the central
limit theorem and Slutsky's lemma,
\[ \sqrt{n}(\hat{\theta} - \theta_0) = \frac{1}{\sqrt{n}} \sum_{i=1}^n \phi_0(W_i) + o_p(1) \Rightarrow N(0, V). \]
\end{proof}

\subsubsection{\texorpdfstring{Proof of
Corollary~\ref{cor-iv-consistency}}{Proof of Corollary~}}\label{proof-of-cor-iv-consistency}

\begin{proof}
Consistency follows from standard Z-estimation arguments. Under B4,
\(\theta_0\) is the unique solution to
\(\mathbb{E}[\phi(W; \theta)] = 0\). The sample analog \(\hat{\theta}\)
solves \(\mathbb{P}_n[\phi(W; \theta, \hat{\eta})] = 0\). By the uniform
law of large numbers and continuity of the moment function,
\(\hat{\theta} \xrightarrow{p} \theta_0\).
\end{proof}

\subsubsection{Causal identification}\label{causal-identification}

We now connect the control function estimand to the causal
semi-elasticity defined through potential outcomes. Decompose the
potential outcomes as \[ \log Y_i(x) = \beta x + M_i(x) + C_i, \] where
\(M_i(x)\) is a potential mediator and \(C_i\) is a confounder invariant
to intervention on \(x\). The causal arithmetic mean semi-elasticity is
\[ \theta(x) = \frac{\partial \log \mathbb{E}[Y_i(x)]}{\partial x} = \beta + \frac{\partial_x \mathbb{E}[\exp(M_i(x))]}{\mathbb{E}[\exp(M_i(x))]}. \]

As established in
\hyperref[myprp-iv-impossibility]{Proposition~\ref*{myprp-iv-impossibility}},
standard instrumental variables approaches cannot identify \(\theta(x)\)
under mean-independence assumptions alone. However,
\hyperref[myprp-control-function-identification]{Proposition~\ref*{myprp-control-function-identification}}
shows that control function assumptions are sufficient. The following
corollary formalizes this for our estimator.

\begin{corollary}[Control function estimand equals causal
semi-elasticity]\protect\hypertarget{cor-cf-causal-identification}{}\label{cor-cf-causal-identification}

Assume the triangular structure \(X = g_0(Z) + V\) with \(Z \perp V\),
and suppose (i) \(C_i \perp X \mid V\) and (ii) \(M_i(x) \perp V\) for
all \(x\). Then the control function estimand
\[ \theta_0 = \mathbb{E}\left[\beta + \frac{\nabla_x \mu_0(X)}{\mu_0(X)}\right], \qquad \mu_0(x) = \mathbb{E}_V[m_0(x, V)], \]
equals the average causal arithmetic mean semi-elasticity
\(\mathbb{E}[\theta(X)]\).

\end{corollary}

\section{Algorithms}\label{algorithms}

\subsection{Ordinary Algorithm: Debiased Estimation of Average
Semi-Elasticity}\label{ordinary-algorithm-debiased-estimation-of-average-semi-elasticity}

\textbf{Input:} Data \(\{(Y_i, X_i, Z_i)\}_{i=1}^n\), number of folds
\(K\)

\textbf{Output:} Estimate \(\hat{\theta}\) and standard error
\(\widehat{\text{SE}}\)

\textbf{Step 0: Sample Splitting}

Randomly partition \(\{1, \ldots, n\}\) into \(K\) folds of
approximately equal size. Let \(k(i)\) denote the fold containing
observation \(i\), and let \(\mathcal{I}_k\) denote the indices in fold
\(k\).

\textbf{Step 1: Cross-Fitted Nuisance Estimation}

For each fold \(k = 1, \ldots, K\):

\begin{quote}
\textbf{(a)} Using data excluding fold \(k\), estimate
\((\hat{\beta}_{-k}, \hat{\gamma}_{-k})\) from the regression
\(\log Y = \beta^\top X + \gamma^\top Z + u\).

\textbf{(b)} Form residuals
\(\hat{u}_i = \log Y_i - \hat{\beta}_{-k}^\top X_i - \hat{\gamma}_{-k}^\top Z_i\)
for \(i \notin \mathcal{I}_k\).

\textbf{(c)} Estimate
\(\hat{m}_{-k}(x,z) \approx \mathbb{E}[e^u \mid X=x, Z=z]\) by
regressing \(e^{\hat{u}}\) on \((X, Z)\) using a neural network on data
excluding fold \(k\).

\textbf{(d)} Obtain \(\nabla_x \hat{m}_{-k}(x,z)\) via automatic
differentiation.

\textbf{(e)} Estimate the density score
\(\widehat{\nabla_x \log f}_{-k}(x \mid z) \approx \nabla_x \log f_0(x \mid z)\)
via score matching on data excluding fold \(k\).

\textbf{(f)} Construct the correction weight:
\[\hat{\alpha}_{-k}(x,z) = -\frac{\widehat{\nabla_x \log f}_{-k}(x \mid z)}{\hat{m}_{-k}(x,z)}\]
\end{quote}

\textbf{Step 2: Compute Score Contributions}

For each observation \(i = 1, \ldots, n\), let \(k = k(i)\) and compute:

\[
\hat{\phi}_i = \hat{\beta}_{-k} + \frac{\nabla_x \hat{m}_{-k}(X_i, Z_i)}{\hat{m}_{-k}(X_i, Z_i)} + \hat{\alpha}_{-k}(X_i, Z_i) \left( e^{\hat{u}_i} - \hat{m}_{-k}(X_i, Z_i) \right)
\]

where
\(\hat{u}_i = \log Y_i - \hat{\beta}_{-k}^\top X_i - \hat{\gamma}_{-k}^\top Z_i\).

\textbf{Step 3: Estimate \(\theta\)}

\[
\hat{\theta} = \frac{1}{n} \sum_{i=1}^n \hat{\phi}_i
\]

\textbf{Step 4: Variance Estimation and Inference}

\[
\hat{V} = \frac{1}{n} \sum_{i=1}^n \left( \hat{\phi}_i - \hat{\theta} \right)^2
\]

\[
\widehat{\text{SE}} = \sqrt{\hat{V}/n}
\]

A \((1-\alpha)\) confidence interval is given by
\(\hat{\theta} \pm z_{1-\alpha/2} \cdot \widehat{\text{SE}}\).

These estimators are valid if \(X\) is as good as randomized. For other
settings, we provide an instrumental variables estimator.

\subsection{IV Algorithm: Debiased Estimation of Average Structural
Semi-Elasticity}\label{iv-algorithm-debiased-estimation-of-average-structural-semi-elasticity}

\textbf{Input:} Data \({(Y_i, X_i, Z_i)}_{i=1}^n\), number of folds
\(K\)

\textbf{Output:} Estimate \(\hat{\theta}\) and standard error
\(\widehat{\text{SE}}\)

\textbf{Sample Splitting.} Randomly partition \({1, \ldots, n}\) into
\(K\) folds. Let \(k(i)\) denote the fold containing observation \(i\).

\textbf{Cross-Fitted Nuisance Estimation.} For each fold
\(k = 1, \ldots, K\), using data excluding fold \(k\):

\begin{quote}
Estimate \(\hat{g}_{-k}(z) \approx \mathbb{E}[X \mid Z = z]\) and form
residuals \(\hat{V}_i = X_i - \hat{g}_{-k}(Z_i)\) for all \(i\).

Estimate \((\hat{\beta}_{-k}, \hat{\rho}_{-k})\) from
\(\log Y = \beta^\top X + \rho^\top \hat{V} + \epsilon\) and form
outcome residuals
\(\hat{u}_i = \log Y_i - \hat{\beta}_{-k}^\top X_i - \hat{\rho}_{-k}^\top \hat{V}_i\).

Estimate
\(\hat{m}_{-k}(x, v) \approx \mathbb{E}[e^u \mid X = x, V = v]\) via
neural network.

Compute \(\hat{\mu}_{-k}(x) = n^{-1} \sum_j \hat{m}_{-k}(x, \hat{V}_j)\)
and \(\nabla_x \hat{\mu}_{-k}(x)\) via automatic differentiation.

Estimate \(\hat{\omega}_{-k}(x, v)\) and \(\widehat{S_X}_{-k}(x)\) via
density ratio estimation and score matching.

Estimate \(\hat{\lambda}_{-k}(z)\) via automatic Riesz representer
learning.
\end{quote}

\textbf{Score Contributions.} For each observation \(i = 1, \ldots, n\),
let \(k = k(i)\) and compute
\[ \hat{\phi}_i = \hat{\beta}_{-k} + \frac{\nabla_x \hat{\mu}_{-k}(X_i)}{\hat{\mu}_{-k}(X_i)} - \frac{\hat{\omega}_{-k}(X_i, \hat{V}_i) \widehat{S_X}_{-k}(X_i)}{\hat{\mu}_{-k}(X_i)} \left( e^{\hat{u}_i} - \hat{m}_{-k}(X_i, \hat{V}_i) \right) - \hat{\lambda}_{-k}(Z_i)^\top (X_i - \hat{g}_{-k}(Z_i)). \]

\textbf{Estimation and Inference.} The estimator is
\[ \hat{\theta} = \frac{1}{n} \sum_{i=1}^n \hat{\phi}_i, \] with
variance estimator and standard error
\[ \hat{V} = \frac{1}{n} \sum_{i=1}^n \left( \hat{\phi}_i - \hat{\theta} \right)^2, \qquad \widehat{\text{SE}} = \sqrt{\hat{V}/n}. \]
A \((1-\alpha)\) confidence interval is given by
\(\hat{\theta} \pm z_{1-\alpha/2} \cdot \widehat{\text{SE}}\).

\end{document}